\documentclass[journal,singlecoloumn]{IEEEtran}
\usepackage[dvipsnames,svgnames,rgb]{xcolor}

\usepackage{amsmath, amssymb, amsthm, amsfonts, mathtools}

\usepackage{enumitem}
\usepackage{hyperref}
\usepackage{tikz}
\usetikzlibrary{arrows.meta, positioning, shapes.misc}
\usetikzlibrary{decorations.markings}
    \tikzset{cross/.style={cross out, draw=black, minimum
        size=2*(#1-\pgflinewidth), inner sep=0pt, outer sep=0pt},
cross/.default={3pt}}

\tikzset{ball/.style={circle, draw, fill=black,inner sep=0pt, minimum width=4pt}}
\tikzset{nd/.style={inner sep=1pt}}
\tikzset{>=Latex}
\tikzset{
  set arrow inside/.code={\pgfqkeys{/tikz/arrow inside}{#1}},
  set arrow inside={end/.initial=>, opt/.initial=},
  /pgf/decoration/Mark/.style={
    mark/.expanded=at position #1 with
    {
      \noexpand\arrow[\pgfkeysvalueof{/tikz/arrow inside/opt}]{\pgfkeysvalueof{/tikz/arrow inside/end}}
    }
  },
  arrow inside/.style 2 args={
    set arrow inside={#1},
    postaction={
      decorate,decoration={
        markings,Mark/.list={#2}
      }
    }
  },
}

\usepackage[disable]{todonotes}
\newcommand{\Hb}{H_b}
\newcommand{\Hsum}{H_{\mathrm{sum}}}
\newcommand{\W}{W}
\newcommand{\Wstar}{W^*}
\newcommand{\eps}{\epsilon}
\newcommand{\K}{\mathbf{K}}
\newcommand{\ki}{k}
\newcommand{\Ptrue}{P_{\mathrm{true}}}
\newcommand{\Pnom}{P_{\mathrm{nom}}}
\newcommand{\Dsum}{D_{\mathrm{sum}}}
\newcommand{\Nsum}{N_{\mathrm{sum}}}
\newcommand{\CRHP}{\ensuremath{\overline{\mathbb{C}}_{+}}}
\newcommand{\OLHP}{\ensuremath{\mathbb{C}_{-}}}

\newcommand\J{\mathrm j}
\newcommand\abs[1]{\ensuremath{\left\lvert#1\right\rvert}}

\DeclareMathOperator\diag{diag}

\DeclareMathOperator\Res{Res}
\DeclareMathOperator\Z{Z}
\DeclareMathOperator\len{len}

\newcommand{\absv}[1]{\lvert #1\rvert}

\newtheorem{theorem}{Theorem}[section]
\newtheorem{lemma}[theorem]{Lemma}
\newtheorem{proposition}[theorem]{Proposition}
\newtheorem{definition}[theorem]{Definition}

\newtheorem{assumption}[theorem]{Assumption}
\newtheorem{remark}[theorem]{Remark}

\newcommand\GB[1]{{\color{black}#1}}
\newcommand\rev[1]{{\color{black}#1}}
\makeatletter
\xdef\@endgadget#1{{\unskip\nobreak\hfil\penalty50\hskip1em\hbox{}\nobreak\hfil#1\parfillskip=0pt\finalhyphendemerits=0\par}}
\newcommand\@Endofsymbol{$\triangledown$}
\newcommand\Endofremark{\@endgadget{\@Endofsymbol}}
\makeatother
\title{On the NP-Hardness of Unconstrained Static Output Feedback Stabilization}
\author{}
\date{\today}

\definecolor{subsectioncolor}{rgb}{0,0,0}
\definecolor{mblue}{rgb}{0,0,0}
\definecolor{nblue}{rgb}{0,0,0}
\makeatletter\let\ps@titlepagestyle\ps@headings\makeatother
\begin{document}
\author{Gal Barkai$^1$, \IEEEmembership{Member, IEEE}, Iman Shames$^2$, \IEEEmembership{Member, IEEE} 
\thanks{$^1$Universit\'e de Lorraine, CNRS, CRAN, email: \textsl{gal.barkai@univ-lorraine.fr}.}
\thanks{$^2$University of Melbourne, Melbourne, Australia, email: \textsl{iman.shames@unimelb.edu.au}.}
}
\maketitle

\begin{abstract}
In this paper we give a polynomial-time reduction from the Subset Sum Problem to unconstrained static output feedback stabilization, establishing its NP-hardness. The construction uses two  scalar building blocks to impose approximate discrete choices and a weighted-sum constraint.  Once the problem is encoded, we relate the stability of the $N+1$ decoupled loops encoding the problem to that of a coupled plant. This is accomplished by leveraging frequency separation via a band-pass transformation, and comparing the root counts of the true characteristic polynomial with that of the different blocks in different regions of the right half-plane. The construction then naturally bounds every stabilizing gain and gives explicit modulus-margin estimates, ensuring that the plant data have polynomial binary encoding length.
\end{abstract}
\section{Introduction}\label{sec:intro}
Static output feedback stabilization is concerned with the question whether a static gain matrix renders a given linear time-invariant plant internally stable. The problem has been studied since the early 1960s; see the survey~\cite{SADG:97}. It is decidable, since the existence of a stabilizing gain is a first-order sentence over the reals and quantifier elimination applies, \cite[Sec.~IV]{ABJ:75} and \cite[Sec.~3.4]{BT:00}. The complexity of related design problems has been the target of investigation for the last few decades.  Checking robust nonsingularity of interval matrices \cite{PR:93} in addition to several robust stability analysis \cite{Nem:93} problems are shown to be NP-hard. Deciding feasibility of bilinear matrix inequalities was proved to be NP-hard in \cite{TO:95}. While, static output feedback stabilizability is equivalent to the feasibility of a bilinear matrix inequality, the hardness of the general bilinear matrix inequality problem does not address the complexity of the stabilization problem itself. NP-hardness of simultaneous stabilization by static output feedback was established in \cite{TO:95} and \cite{BT:97}. In \cite{BT:97}, NP-hardness of stabilization by static state or output feedback when the entries of the gain matrix are constrained to given intervals is proved. Later, NP-hardness of static output feedback pole placement was proved in~\cite{Fu:04}. Pole placement prescribes closed-loop eigenvalues exactly, while stabilization requires only that all eigenvalues lie in the open left half-plane. In \cite{Fu:04} it was stated that neither problem reduces to the other in an obvious way. The unconstrained stabilization problem remained open. It was posed as Open Problem~1 in~\cite[Ch.~11]{OpenProb:99}, and it was conjectured in \cite[Sec.~1]{BT:97} that it is NP-hard.  We prove its NP-hardness in this paper.

Our reduction is motivated by two main ideas: i) the hardness argument of \cite{BT:97} gains with interval constraints, and ii) frequency domain notions such a frequency separation. Given an instance with $N$ weights, we construct a rational plant with one input, $N$ outputs and $10N+3$ states. Each output contains a band-pass copy of a scalar prototype whose stabilizing gains lie in two narrow intervals around $1$ and $2$, together with a weighted copy of a shared summing block. The latter imposes a narrow constraint on $\W=\sum_i a_i\ki_i$ in the comparison system. The interval widths leave enough slack to round a feasible real gain to a subset-sum solution. No gain constraint is imposed on the SOFS instance; the bounds needed in the analysis must follow from stability of the constructed plant. The scalar loops are coupled in the plant, so their stability conditions cannot simply be combined. We separate their frequency ranges and compare the true characteristic polynomial with one scalar factor at a time, retaining the other open-loop denominators. On each comparison contour, the selected return difference dominates the other loop terms whenever its modulus margin is large enough. Rouch\'e's theorem then preserves the local right half-plane root count within each region. This local comparison is needed in the reverse direction, where some gains may lie near stability crossings and need not have positive margins. Bounds on all stabilizing gains and explicit scalar margin estimates allow the required frequencies to be chosen with polynomial binary encoding length.

Section~\ref{sec:prob} gives the decision problems and rounding argument. Sections~\ref{sec:blocks} and~\ref{sec:equiv} construct the plant and establish the local comparison. Section~\ref{sec:reduction} chooses its parameters and proves both directions of the reduction; the appendices contain the prototype calculations and frequency-separation estimates.
\begin{remark}[Novelty]
The construction in this paper has been developing for approximately a year, during which the question of NP hardness has remained open. During the final revising stages of this manuscript, we have been made aware of two AI-assisted preprints addressing the same problem appearing in recent weeks, with \cite{Lof:26} proposing a reduction from Betweenness, while \cite{ACNT:26} claims strong hardness via reduction from 3SAT.  \Endofremark
\end{remark}
%
\section{Problem statements and the continuous relaxation}\label{sec:prob}
\rev{Write $\CRHP=\{s\in\mathbb C:\operatorname{Re}s\ge0\}$ and
$\OLHP=\{s\in\mathbb C:\operatorname{Re}s<0\}$. A polynomial is Hurwitz if all its roots lie in $\OLHP$, and a square matrix is Hurwitz stable if all its eigenvalues do. For a polynomial $P$, $\Z(P;S)$ counts its roots in $S$ with multiplicity. We write $\mathbb Z_+$ for the positive integers and $x_i$ for the $i$-th entry of a vector $x$.}

\rev{We also use one algebraic tool to quantify how far scalar characteristic polynomials remain from imaginary-axis zeros. For $\varphi(y)=\sum_{i=0}^m\varphi_i y^i$ and $\psi(y)=\sum_{j=0}^n\psi_j y^j$, with $m,n\ge1$ and $\varphi_m\psi_n\ne0$, the \emph{resultant} $\Res_y(\varphi,\psi)$ is the determinant of the Sylvester matrix. Its rows are the coefficient vectors of $y^{n-1}\varphi,\ldots,\varphi,y^{m-1}\psi,\ldots,\psi$ in the basis $y^{m+n-1},\ldots,1$, and it vanishes exactly when $\varphi$ and $\psi$ have a common complex root~\cite[Sec.~4.1, Theorem~4.1]{Stu:02}. The subscript identifies the eliminated variable.
}

\rev{The input data of the following decision problems have finite binary encodings; $\len(\mathcal I)$ denotes the encoding length of an instance $\mathcal I$. Note that the feedback gain is allowed to be real and is not part of the input.}
\begin{definition}[Unconstrained SOFS]\label{def:sofs}
Given $A\in\mathbb Q^{n\times n}$, $B\in\mathbb Q^{n\times m}$ and $C\in\mathbb Q^{\ell\times n}$, decide whether there exists $\K\in\mathbb R^{m\times\ell}$ such that $A+B\K C$ is Hurwitz stable. 
\end{definition}
\begin{definition}[Subset Sum Problem (SSP)]\label{def:ssp}
Given positive integers $a_1,\ldots,a_N$ and $S$, decide whether there exists $x\in\{0,1\}^N$ such that $\sum_i a_ix_i=S$.
\end{definition}
SSP is NP-complete~\cite[Problem SP13]{GJ:79}, and its input length is the total bit length of $S$ and the $a_i$. We establish hardness of SOFS by mapping each SSP instance to a rational plant in polynomial time, preserving both YES and NO answers.
\begin{definition}[Karp reduction]\label{def:karp}
A Karp reduction from a decision problem $\Pi_0$ to a decision problem $\Pi$ is a polynomial-time computable map $f$ such that an instance $\mathcal I$ is a YES instance of $\Pi_0$ if and only if $f(\mathcal I)$ is a YES instance of $\Pi$. Time is measured in the binary input length.
\end{definition}
A problem is in NP if its YES instances have certificates of polynomial length, verifiable in polynomial time. It is \emph{NP-hard} if every problem in NP admits a Karp reduction to it, and \emph{NP-complete} if it is also in NP~\cite{Gol:10}. Since reductions can be composed, a Karp reduction from SSP suffices to prove that SOFS is NP-hard, it does not establish that SOFS is in NP. In our construction, one direction turns a SSP  solution into a stabilizing gain, while the other direction must recover a SSP solution from any stabilizing real gain. The latter direction is essential because no constraints are imposed on $\K$ in Definition~\ref{def:sofs}. Polynomial-time construction also bounds the length of the rational plant data, even when their magnitudes are large.
%
\subsection{The continuous subset sum problem}\label{sec:cssp}
For a single-input plant with $N$ outputs, the gain $\K=(\ki_1,\ldots,\ki_N)$ is a row vector. To connect its real entries with the discrete choices in SSP, we replace those choices by narrow open intervals and retain enough slack to round back to an exact solution. This is the continuous relaxation of SSP.
\begin{definition}[Continuous SSP (CSSP)]\label{def:cssp}
Given $a\in\mathbb Z_+^N$, $\Wstar\in\mathbb Z_+$ and rational $\eps,\delta\in(0,1)$, decide whether there exists $\K\in\mathbb R^{1\times N}$ with $\ki_i\in(1-\eps,1+\eps)\cup(2-\eps,2+\eps)$ for all $i$ and
$\abs{\sum_ia_i\ki_i-\Wstar}<\delta$.
\end{definition}
\begin{lemma}[Rounding for CSSP]\label{lem:cssp}
Let $(a,\Wstar,\eps,\delta)$ be a CSSP instance with $\delta+\eps\sum_i a_i<1$. It is a YES instance if and only if there exists $\hat\K\in\{1,2\}^{1\times N}$ with $\sum_i a_i\hat\ki_i=\Wstar$.
\end{lemma}
\begin{proof}
First, note that such a $\hat\K$ itself satisfies the CSSP constraints. Conversely, for any feasible $\K$, choose $\hat\ki_i\in\{1,2\}$ with $\abs{\ki_i-\hat\ki_i}<\eps$. Then
\[
\begin{aligned}
\abs{\sum_i a_i\hat\ki_i-\Wstar}
&\leq \abs{\sum_i a_i\ki_i-\Wstar}
 +\sum_i a_i\abs{\hat\ki_i-\ki_i}\\
&<\delta+\eps\sum_i a_i<1.
\end{aligned}
\]
Since the expression inside the absolute value is an integer, it is zero.
\end{proof}
For a given SSP instance, set $\Wstar=S+\sum_i a_i$, $\delta=1/4$ and $\eps=1/(8\sum_i a_i)$. These rational parameters have polynomial bit length and satisfy $\delta+\eps\sum_i a_i=3/8<1$. With $x_i=\hat\ki_i-1$, Lemma~\ref{lem:cssp} makes the two instances equivalent. in the sense that the CSSP instance is a YES instance if and only if the SSP instance is a YES instance. The shift to $\{1,2\}$ matches the two stabilizing gain intervals of the scalar prototype in Section~\ref{sec:blocks}. This rounding argument adapts~\cite[Lemma~1]{TO:95} to the shifted target and the stated slack.

\section{Scalar building blocks and the plant family}\label{sec:blocks}
As discussed in Sec.~\ref{sec:intro}, previous NP-hardness results, for example~\cite{BT:97}, establish hardness of stabilization problems in which the feedback gain is subject to prescribed interval constraints. Inspired by these constructions, we seek to construct a family of \emph{plants} that enforce similar constraints through the requirement of closed-loop stability. The gain $\K$ is otherwise unconstrained, so these restrictions must hold for every stabilizing gain, rather than being imposed as part of the decision problem.

To this end, we construct the plant family out of two particular scalar building blocks. First, we construct a scalar prototype, denoted $\Hb$, whose stabilizing gains form two narrow intervals around $1$ and $2$. The prototype as well as the width of the stability intervals are parametrized by $\epsilon$, one of the two slackness parameters defined for the continuous SSP in \S~\ref{sec:cssp}. Next, we construct a second prototype, $\Hsum$, whose stabilizing gains form a narrow interval around the SSP target $\Wstar$. The prototype is parameterized via $\Wstar$ and $\delta$, the second slackness parameter, the latter also determines the width of the interval. We than introduce a band-pass transformation, allowing us generate frequency-separated copies of $\Hb$ with the same Nyquist plot and stability intervals. Independently closing these scalar loops would impose the desired constraints on the individual gains and their weighted sum.

\subsection{The two-interval prototype}\label{subsec:Hbprot}

We first construct a scalar system whose stabilizing gains form two disjoint intervals around $1$ and $2$. To this end, we prescribe the gains at which the closed-loop poles reach the imaginary axis, and then verify that the intervening stability intervals are the desired ones. For $\eps\in(0,1/8]$, write
\[
\begin{gathered}
D_b(p)=p^5+\sum_{j=0}^4d_jp^j,\quad 
N_b(p)=-p^3+\sum_{j=0}^2n_jp^j,\\
P_\ki\coloneqq D_b+\ki N_b.
\end{gathered}
\]
We choose the eight coefficients so that
\[
\begin{gathered}
P_{1-\eps}(0)=0, \quad P_{1+\eps}(\J)=0,\quad P_{2-\eps}(\J\sqrt2)=0,\\ 
P_{2+\eps}(\J\sqrt3)=0,\quad 
N_b(3)=0.
\end{gathered}
\]
Since the coefficients are real, these are eight real linear conditions. The first four prescribe imaginary-axis roots at the endpoints of the desired stabilizing intervals, while the last places a zero in the open right half-plane.  The interpolation conditions have a unique solution for every $\eps\in(0,1/8]$, and we define $\Hb\coloneqq N_b/D_b$ using this solution. For rational $\eps$, its coefficients are rational and can be computed in polynomial time with bit length polynomial in that of $\eps$. The explicit construction is given in Appendix~\ref{app:Hb}, together with the verification that the prescribed crossings yield exactly the required stability intervals.
\begin{proposition}[Prototype $\Hb$]\label{prop:Hb}
Let $\eps\in(0,1/8]$. The prototype constructed above has the following properties.
\begin{enumerate}[label=(\roman*)]
\item\label{propitem:hurwitz2} For every $\ki\in\mathbb{R}$, the polynomial $P_\ki$ is Hurwitz if and only if
\[
\ki\in(1-\eps,1+\eps)\cup(2-\eps,2+\eps).
\]
Moreover, $P_\ki$ has a root on the imaginary axis if and only if $\ki\in\{1-\eps,1+\eps,2-\eps,2+\eps\}$.

\item \label{propitem:nothurwitz} The denominator $D_b$ has exactly one root in $\CRHP$; it is real and lies in $(0,1)$. In particular, $D_b$ is not Hurwitz and has no roots on the imaginary axis.

\item\label{propitem:hurwitz} The numerator has the factorization
\[
\begin{gathered}
N_b(p)=-(p-3)(p^2+\zeta_1p+\zeta_0),\\
\zeta_1\coloneqq \frac{4\eps+5}{7(1-4\eps)},\qquad
\zeta_0\coloneqq \frac{4(2-11\eps)}{7(1-4\eps)}.
\end{gathered}
\]
Since $\zeta_1,\zeta_0>0$, the quadratic factor is Hurwitz. Thus $p=3$ is the only root of $N_b$ in $\CRHP$, and it is simple.

\item\label{propitem:noroots} The denominator has no roots in $\{p:\abs{p-3}\leq7/8\}$, and
\[
\abs{\Hb(p)}\ge\frac1{500} \qquad\text{whenever}\qquad \frac{5}{8}\leq\abs{p-3}\leq\frac{7}{8}.
\]
\item\label{propitem:coeff} The coefficients satisfy
\[
\begin{gathered}
0<d_4\leq\frac{6}{7},\quad
 7\leq d_3\leq 9,\quad 
 S_d\coloneqq \sum_{j=0}^4\abs{d_j}\leq 19,\\
 S_n\coloneqq \sum_{j=0}^2\abs{n_j}\leq 10.
\end{gathered}
\]
The individual coefficient bounds are given in Appendix~\ref{app:Hb}.
\end{enumerate}
\end{proposition}
Note that only the first property encodes the two choices of a binary variable and that there are many possible other systems satisfying this property. The reason for the particular solution choice is to obtain the rest of the properties, is to facilitate the reduction, and will be explained in detail in the following sections. 
\subsection{The summing block}
The following proposition constructs the block that constrains the weighted gain
sum.
\begin{proposition}[Summing block $\Hsum$]\label{prop:hsum}
Let $\Wstar\in\mathbb{Z}_{+}$ and $\delta\in(0,1)$, and set
$q_0\coloneqq (\Wstar)^2-\delta^2>0$.vDefine
\[
\Hsum(s)\coloneqq \frac{s^2-s}{s^3+2\Wstar s+q_0}=:\frac{\Nsum(s)}{\Dsum(s)} .
\]
Then $\Dsum$ has no roots on the imaginary axis, and for $\W\in\mathbb{R}$ the
polynomial
\[
\begin{aligned}
\chi_{\mathrm{sum}}(s,\W)&\coloneqq \Dsum(s)+\W\Nsum(s)\\
&=s^3+\W s^2+(2\Wstar-\W)s+q_0
\end{aligned}
\]
is Hurwitz if and only if $\W\in(\Wstar-\delta,\,\Wstar+\delta)$.
Moreover $\chi_{\mathrm{sum}}(s,\W)$ has a root on the imaginary axis if and
only if $\W\in\{\Wstar-\delta,\,\Wstar+\delta\}$.
\end{proposition}

\begin{proof}
A monic cubic $s^3+a_2s^2+a_1s+a_0$ is Hurwitz if and only if $a_2>0$, $a_0>0$,
and $a_2a_1>a_0$.
Here $a_2=\W$, $a_1=2\Wstar-\W$, $a_0=q_0>0$, and
\[
a_2a_1-a_0=\delta^2-(\W-\Wstar)^2 .
\]
Let $\abs{\W-\Wstar}<\delta$.
Then $\W>\Wstar-\delta>0$ and $2\Wstar-\W>\Wstar-\delta>0$, because
$\delta<1\le\Wstar$.
All three conditions hold.
If $\abs{\W-\Wstar}\ge\delta$ then $a_2a_1\le a_0$ and the polynomial is not
Hurwitz.

Axis roots: $\chi_{\mathrm{sum}}(0,\W)=q_0\neq0$, so no root at the origin.
For $\nu\neq0$,
$\chi_{\mathrm{sum}}(\J\nu,\W)=\bigl(q_0-\W\nu^2\bigr)+\J\nu\bigl(2\Wstar-\W-\nu^2\bigr)$.
Both parts vanish exactly when $\nu^2=2\Wstar-\W>0$ and $\W(2\Wstar-\W)=q_0$,
that is $(\W-\Wstar)^2=\delta^2$.
Both conditions then hold, since $2\Wstar-\W=\Wstar\mp\delta>0$.
At $\W=0$ the same formula shows $\Dsum(\J\nu)\neq0$ for all $\nu$, so $\Dsum$
has no roots on the imaginary axis.
\end{proof}
%
\subsection{The band-pass transformation}\label{sec:bp}
For prescribed $B_w>0$ and center frequency $\omega>0$, define
\begin{equation}\label{eq:pw}
    p_\omega(s)\coloneqq \frac{s^2+\omega^2}{B_w s},
\qquad
\Hb^{\omega}(s)\coloneqq \Hb\bigl(p_\omega(s)\bigr).
\end{equation}
The map $s\mapsto p_\omega(s)$ is the low-pass to band-pass transformation of
analog filter design; see~\rev{\cite[Sec.~7.3.4, eq.~(7.42)]{CA:19book}}. 
Thus, $\Hb^{\omega}=N_{\omega}/D_{\omega}$ with
\begin{equation}\label{eq:NiDi}\begin{gathered}
D_{\omega}(s)\coloneqq \sum_{j=0}^{5}d_j\,(s^2+\omega^2)^j\,(B_ws)^{5-j},
\\
N_{\omega}(s)\coloneqq \sum_{j=0}^{3}n_j\,(s^2+\omega^2)^j\,(B_ws)^{5-j},
\end{gathered}\end{equation}
where $d_5\coloneqq 1$ and $n_3\coloneqq -1$. The polynomial $D_{\omega}$ is monic of degree $10$. The polynomial $N_{\omega}$ has degree $8$ and is divisible by $s^2$, while $D_{\omega}(0)=\omega^{10}$. The blocks of the plant are $\Hb^i\coloneqq \Hb^{\omega_i}$ with $N_i\coloneqq N_{\omega_i}$ and $D_i\coloneqq D_{\omega_i}$. We further define the following modulus margin, which is required for the arguments of Section~\ref{sec:equiv}.
\begin{definition}[Margin]\label{def:margin}
For a real rational strictly proper $F$ with no poles on the imaginary axis and
a gain $g\in\mathbb{R}$, set
\[
m_F(g)\coloneqq \inf_{\nu\in\mathbb{R}}\abs{1+gF(\J\nu)} .
\]
Write $m_b(\ki)\coloneqq m_{\Hb}(\ki)$ and $m_s(\W)\coloneqq m_{\Hsum}(\W)$.
\end{definition}
The following proposition collects the salient properties of the transformation
that will be used later in the reduction.
\begin{proposition}[Band-pass block]\label{prop:bpmap}
Let $\eps\in(0,1/8]$, $B_w>0$, $\omega>0$, and $\ki\in\mathbb{R}$.
\begin{enumerate}[label=(\roman*)]
\item For $s\neq0$,
$\operatorname{Re}p_\omega(s)=\frac{1}{B_w}\operatorname{Re}(s)\bigl(1+\omega^2/\abs{s}^2\bigr)$. Hence $p_\omega$ maps the open right half-plane to itself, the open left half-plane to itself, and the set $\{\J \nu: \nu\in\mathbb{R},\nu\neq 0\}$ to the imaginary axis. \label{propitem:bpmap1}
\item The 10 roots of $D_{\omega}+\ki N_{\omega}$ are exactly the solutions $s$ of
$s^2-B_wqs+\omega^2=0$, where $q$ ranges over the roots of $D_b+\ki N_b$. Consequently $D_{\omega}+\ki N_{\omega}$ is Hurwitz if and only if $D_b+\ki N_b$
is Hurwitz.\label{propitem:bpmap2}
\item $\inf_{\omega'}\abs{1+\ki \Hb^{\omega}(\J\omega')}=m_b(\ki)$. \label{propitem:bpmap3}
\item Let $(A_b,B_b,C_b)$ be the companion realization of $\Hb$ in controllable canonical form, of dimension $5$. Then
\[
\begin{gathered}
A_\omega=\begin{bmatrix}0&\omega I\\-\omega I&B_wA_b\end{bmatrix}, \quad 
B_\omega=\begin{bmatrix}0\\B_b\end{bmatrix},\\
C_\omega=\begin{bmatrix}0&B_wC_b\end{bmatrix}
\end{gathered}
\]
is a realization of $\Hb^{\omega}$ of dimension $10$, and $\det(sI-A_\omega)=D_{\omega}(s)$. \label{propitem:bpmap4}
\end{enumerate}
\end{proposition}
\begin{proof}
    \ref{propitem:bpmap1} It is a consequence of direct computation from $p_\omega(s)=\bigl(s+\omega^2/s\bigr)/B_w$ and
$\operatorname{Re}(1/s)=\operatorname{Re}(s)/\abs{s}^2$.

\ref{propitem:bpmap2}
\GB{Let $q_1,\dots,q_5$ be the roots of $D_b+\ki N_b$, counted with multiplicity. The polynomial identity
\[
D_{\omega}(s)+\ki N_{\omega}(s)
=\prod_{r=1}^{5}\bigl(s^2-B_wq_rs+\omega^2\bigr)
\]
shows that its roots are the corresponding quadratic roots.  Zero is not a root, becasue at $s=0$ the product equals $\omega^{10}$.} If $q\in\OLHP$, both solutions of $s^2-B_w q s+\omega^2=0$ lie in $\OLHP$, because a solution in $\CRHP\setminus\{0\}$ would give $p_\omega(s)\in\CRHP$ by \ref{propitem:bpmap1}, and $s=0$ is not a solution. If $q\in\CRHP$, both solutions lie in $\CRHP$ by the same argument applied to the left half-plane. Hence $D_{\omega}+\ki N_{\omega}$ is Hurwitz if and only if $D_b+\ki N_b$ is Hurwitz.

\ref{propitem:bpmap3}
 For $\omega'>0$, $p_\omega(\J\omega')=\J\Omega(\omega')$ with
$\Omega(\omega')=\bigl(\omega'-\omega^2/\omega'\bigr)/B_w$.
The map $\Omega$ increases strictly from $-\infty$ to $+\infty$ on $(0,\infty)$.
Hence
$\{\abs{1+\ki \Hb^{\omega}(\J\omega')}:\omega'>0\}=\{\abs{1+\ki \Hb(\J\Omega)}:\Omega\in\mathbb{R}\}$.
The value at $\omega'=0$ is $1$, and the right-hand set contains values
arbitrarily close to $1$, since $\abs{1+\ki\Hb(\J\Omega)}\to1$ as
$\abs\Omega\to\infty$.
Including this value therefore does not change the infimum.
Negative $\omega'$ give the same set.
The infima agree.

\ref{propitem:bpmap4}
In the Laplace domain the state equations are $sz_1=\omega z_2$ and
$sz_2=-\omega z_1+B_wA_bz_2+B_bu$.
Eliminating $z_1$ gives $p_\omega(s)z_2=A_bz_2+\frac{1}{B_w}B_bu$, hence
$y=B_wC_bz_2=C_b\bigl(p_\omega(s)I-A_b\bigr)^{-1}B_bu=\Hb\bigl(p_\omega(s)\bigr)u$,
since $C_b(pI-A_b)^{-1}B_b=\Hb(p)$ for the companion realization.
For $s\neq0$, the Schur complement of the block $sI$ in $sI-A_\omega$ gives
\begin{align*}
\det(sI-A_\omega)&=\det(sI)\,\det\Bigl(sI-B_wA_b+\frac{\omega^2}{s}I\Bigr)\\&=(B_ws)^5\det\bigl(p_\omega(s)I-A_b\bigr)\\
&=(B_ws)^5D_b\bigl(p_\omega(s)\bigr)=D_{\omega}(s),
\end{align*}
since the characteristic polynomial of the companion matrix $A_b$ is $D_b$, and
the last equality is \eqref{eq:NiDi}.
Both sides are polynomials in $s$ that agree for $s\neq0$, hence for all $s$.
\end{proof}
\todo[inline,author=Iman,color=cyan]{This paragraph is hard to parse for someone who does not know the proof approach. E.g., what separation? What disjoint stability intervals? etc. Do we need it here?}
\todo[inline, author=Gal]{Nope. I'm moving a lot of exposition and explanations to later on now.}

\begin{remark}
    Parts of Proposition \ref{prop:bpmap} can be shown via different control theoretic arguments. In particular, it can be shown that under the frequency transformation the frequency response of $\Hb^\omega$ traverses the full Nyquist locus of $\Hb$ twice with the same orientation. Strict properness gives the common limiting value zero at the ends of these traversals. Due to the doubling of the poles, this implies the stability regions of $\Hb^\omega$ and $\Hb$ are the same. Item \ref{propitem:bpmap3} can be deduced using the same arguments, since it is equivalent to the shortest Euclidean distance from the Nyquist plot and the critical point. \Endofremark
\end{remark}
\subsection{The plant family}\label{sec:plant}
%
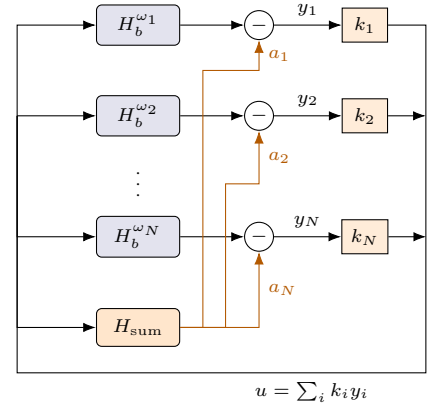
\begin{figure}[!hbt]
\centering
 \begin{tikzpicture}[>=Latex, font=\scriptsize,
          blk/.style={draw, rounded corners=2pt, minimum width=1.1cm,
            minimum height=0.5cm, fill=MidnightBlue!12},
          gain/.style={draw, fill=orange!15, minimum width=0.6cm,
            minimum height=0.45cm},
          sum/.style={draw, circle, inner sep=1pt, minimum size=0.35cm}]
        \coordinate (u) at (-1.6,0.4);
        \foreach \i/\yy in {1/1.6, 2/0.4, N/-1.2}{
          \node[blk] (H\i) at (0,\yy) {$\Hb^{\omega_{\i}}$};
          \node[sum] (s\i) at (1.6,\yy) {$-$};
          \node[gain] (k\i) at (3.0,\yy) {$\ki_{\i}$};
          \draw[->] (H\i) -- (s\i);
          \draw[->] (s\i) -- (k\i) node[midway,above]{$y_{\i}$};
          \draw[->] (u) |- (H\i);
        }
        \node at (0,-0.4) {$\vdots$};
        \node[blk, fill=orange!20] (Hs) at (0,-2.4) {$\Hsum$};
        \draw[->] (u) |- (Hs);
        \draw[->, orange!70!black] (Hs.east) -| node[pos=0.75,
            right]{$a_{N}$} (sN.south);
        \draw[->, orange!70!black] (Hs.east) --++ (+.6,0) --++ (0,1.9) -| node[pos=0.75,
            right]{$a_{2}$} (s2.south);
        \draw[->, orange!70!black] (Hs.east) --++ (.3,0) --++ (0,3.4) -| node[pos=0.75,
            right]{$a_{1}$} (s1.south);
        \draw (k1.east) -- ++(0.5,0) |- (2.5,-3.0) -| (u);
        \draw[->] (k2.east) -- ++(0.5,0);
        \draw[->] (kN.east) -- ++(0.5,0);
        \node[font=\scriptsize] at (2.3,-3.25) {$u=\sum_i \ki_i y_i$};
      \end{tikzpicture}\par
\caption{The plant structure: a parallel interconnection of $N$ band-pass copies of $\Hb$ each centered at a different $\omega_i$, plus $a_i$ times $\Hsum$. }
\label{fig:plantstruct}
\end{figure}
Using the scalar building blocks introduced above, we construct single-input plants, $m=1$, with $\ell=N$ outputs. For such plants the gain is a row vector $\K=(\ki_1,\dots,\ki_N)\in\mathbb{R}^{1\times N}$, minimizing the coupling between components without constraining the gain. The plant transfer function is chosen as
\begin{equation}\label{eq:plantStruct}
G(s)=C(sI-A)^{-1}B=-\begin{bmatrix}\Hb^1(s)+a_1\Hsum(s)\\ \vdots\\ \Hb^N(s)+a_N\Hsum(s)\end{bmatrix},
\end{equation}
where $\Hsum=\Nsum/\Dsum$ is the summing block of Proposition~\ref{prop:hsum} and $\Hb^i=N_i/D_i$ is the band-pass copy of the prototype of Section~\ref{sec:bp} centered at frequency $\omega_i$. The center frequencies $0<\omega_1<\dots<\omega_N$ will be constrained in
Section~\ref{sec:equiv} and fixed in Definition~\ref{def:schedule}. Recall the determinant identity
\begin{equation}\label{eq:sylvester}
\det\bigl(sI-(A+B\K C)\bigr)=\det(sI-A)\,\bigl(1-\K C(sI-A)^{-1}B\bigr),
\end{equation}
which holds because $\det(I-uv^{\mathsf T})=1-v^{\mathsf T}u$ for column vectors $u,v$. Hence, for \eqref{eq:plantStruct} we have
\begin{equation}\label{eq:chitrue}
\begin{gathered}
\begin{aligned}
  \chi(s,\K) &\coloneqq 1-\K G(s)\\&=1+\W\Hsum(s)+\sum_{i=1}^{N}\ki_i\Hb^i(s) ,  
\end{aligned} 
\end{gathered}
\end{equation}
where $\W\coloneqq \sum_{i=1}^N a_i\ki_i$. Let $D_i$ and $\Dsum$ be the denominators of $\Hb^i$ and $\Hsum$ which are monic of degrees $10$ and $3$, respectively. Define $D_{\mathrm{tot}}\coloneqq \Dsum\prod_{j=1}^ND_j$, and the polynomial
\begin{equation}\label{eq:Ptrue}
\begin{aligned}
\Ptrue(s,\K)&\coloneqq D_{\mathrm{tot}}\,\chi(s,\K)\\
&=\Dsum\prod_jD_j+\W\Nsum\prod_jD_j\\
&\quad+\sum_{i=1}^N\ki_iN_i\Dsum\prod_{j\ne i}D_j,
\end{aligned}
\end{equation}
which is monic of degree $d\coloneqq 10N+3$.  Lemma~\ref{lem:realization} realizes \eqref{eq:plantStruct} in state space and identifies $\Ptrue$ as the closed-loop characteristic polynomial with gain $\K$. All stability statements below are made for polynomials.
\begin{lemma}[Realization]\label{lem:realization}
Let $(A_b,B_b,C_b)$ be the companion realization of $\Hb$ in controllable
canonical form, and let $(A_s,B_s,C_s)$ be the companion realization of $\Hsum$
in the same form.

Let $(A_i,B_i,C_i)\coloneqq (A_{\omega_i},B_{\omega_i},C_{\omega_i})$ be the realization
of $\Hb^i$ of dimension $10$ obtained from $(A_b,B_b,C_b)$ in
Proposition~\ref{prop:bpmap}\ref{propitem:bpmap4}.
Define $A\coloneqq \diag(A_1,\dots,A_N,A_s)$, let $B$ stack the input matrices in the
same order, and let row $i$ of $C$ contain $-C_i$ in the block of $A_i$ and
$-a_iC_s$ in the block of $A_s$, all other entries zero.
Then $n=10N+3$, the transfer function is \eqref{eq:plantStruct}, and for every
$\K\in\mathbb{R}^{1\times N}$
\[
\det\bigl(sI-(A+B\K C)\bigr)=\Ptrue(s,\K).
\]
Consequently $A+B\K C$ is Hurwitz stable if and only if $\Ptrue(s,\K)$ is
Hurwitz. \GB{If $\eps$, $\delta$, $B_w$, and the $\omega_i$ are rational, then the entries of $(A,B,C)$ are rational, with bit length polynomial in the bit lengths of these parameters, $\Wstar$, and the $a_i$.}
\end{lemma}

\begin{proof}
The transfer function of the block-diagonal realization is
\eqref{eq:plantStruct} by construction, since $C_i(sI-A_i)^{-1}B_i=\Hb^i$ by
Proposition~\ref{prop:bpmap}\ref{propitem:bpmap4} and $C_s(sI-A_s)^{-1}B_s=\Hsum$.
The characteristic polynomial of $A$ is the product of the block characteristic
polynomials.
These are $\det(sI-A_i)=D_i$ by Proposition~\ref{prop:bpmap}\ref{propitem:bpmap4} and
$\det(sI-A_s)=\Dsum$ for the companion block, so $\det(sI-A)=D_{\mathrm{tot}}$.
Identity \eqref{eq:sylvester} then gives
$\det(sI-A-B\K C)=D_{\mathrm{tot}}\chi(s,\K)=\Ptrue(s,\K)$ as rational
functions.
Both sides are monic polynomials of degree $d$, so they are equal as
polynomials.
The stability equivalence follows.

Bit length: the entries of $(A_b,B_b,C_b)$ and $(A_s,B_s,C_s)$ are the
coefficients of $N_b,D_b$ and of $\Nsum,\Dsum$, which are rationals whose bit
lengths are polynomial in the bit lengths of $\eps$, $\delta$, and $\Wstar$.
The construction of Proposition~\ref{prop:bpmap}\ref{propitem:bpmap4} adds the entries of
$\omega_iI$ and multiplications by $B_w$, and the rows of $C$ multiply $C_s$ by
the integers $a_i$.
\end{proof}

\section{Comparison, frequency separation and local root counts}\label{sec:equiv}
The stability of \eqref{eq:plantStruct} does not seem to immediately encode a solution to CSSP. However, the various blocks \emph{independently} do. To see this, consider the polynomial 
\begin{equation}\label{eq:Pnom}
     \Pnom(s,\K)\coloneqq \bigl(\Dsum+\W\Nsum\bigr)\prod_{i=1}^N\bigl(D_i+\ki_i N_i\bigr)
\end{equation}
which is also monic of degree $d\coloneqq 10N+3$, and is the product of the closed-loop characteristic polynomials of the $N+1$ decoupled scalar loops. By Propositions~\ref{prop:Hb}\ref{propitem:hurwitz2}, \ref{prop:bpmap}\ref{propitem:bpmap2}, and~\ref{prop:hsum}, the polynomial $\Pnom(\cdot,\K)$ is Hurwitz if and only if $\K$ satisfies the two constraints of Definition~\ref{def:cssp} for the instance $(a,\Wstar,\eps,\delta)$. With $\Wstar=S+\sum_i a_i$ and $\delta+\eps\sum_i a_i<1$, Lemma~\ref{lem:cssp} turns any such $\K$ into a solution $x\in\{0,1\}^N$ of the SSP instance $(a,S)$, and conversely every solution $x$ gives the gain $\hat\K=\mathbf{1}+x$ for which $\Pnom$ is Hurwitz.

In this section, we use those individual loops as a \emph{comparison system} with $\Ptrue$, relating the right half-plane roots of $\Pnom$ with those of $\Ptrue$ in different regions. Fix the blocks of Section~\ref{sec:blocks} and write
\[
\begin{gathered}
L_0\coloneqq\W\Hsum,\qquad
L_i\coloneqq\ki_i\Hb^i,\quad 1\le i\le N,\\
\end{gathered}
\]
Then the two characteristic polynomials satisfy
\[
\frac{\Ptrue}{D_{\mathrm{tot}}}=1+\sum_{i=0}^N L_i=\chi,
\qquad
\frac{\Pnom}{D_{\mathrm{tot}}}=\prod_{i=0}^N(1+L_i).
\]
A natural approach would be to compare these polynomials directly on a contour enclosing their right half-plane roots, in a similar fashion to the Nyquist criterion. Frequency separation makes the other loop terms small near each selected band, but this alone does not give the \emph{relative} bound required by Rouch\'e's theorem. To see the difficulty, consider the case $N=1$, for which $\chi-(1+L_0)(1+L_1)=-L_0L_1.$ At a frequency where $L_1=-1$, the nominal product vanishes, while the coupled return difference equals $L_0$. Thus, even an arbitrarily small nonzero tail $L_0$ prevents strict domination by the nominal product at that point. 
\begin{figure}[!hbt]
\centering
 \begin{tikzpicture}[>=Latex, font=\scriptsize, xscale=01.05, yscale=1.35]
        \draw[->] (0,0) -- (6,0) node[right]{$\absv{s}$ (log)};
        \draw[->] (0,0) -- (0,1.7) node[above]{$\absv{L_j(\J\nu)}$};
        \draw[orange!80!black, thick, domain=0.1:5.6, samples=60]
          plot (\x, {0.9/(1+1.6*\x)});
        \node[orange!70!black] at (0.9,0.85) {$L_0=W\Hsum$};
        \foreach \c/\lab in {1.7/{\omega_1}, 3.1/{\omega_2}, 4.5/{\omega_N}}{
          \draw[MidnightBlue, thick, domain={\c-1.2}:{\c+1.2}, samples=50]
            plot (\x, {1.2*exp(-((\x-\c)/0.28)^2)});
          \draw[dashed] (\c,0) -- (\c,-0.12) node[below]{$\lab$};
        }
        \node[MidnightBlue] at (3.1,1.45) {band terms $L_i=\ki_i\Hb^{i}$};
        \fill[red!25] (2.25,0) rectangle (2.55,0.25);
        \node[red!80!black, align=center] at (2.4,0.5) {overlap\\ of tails};
        \node[align=center] at (4.3,-0.55) {summing tail decays like
            $1/\absv{s}$,\\ band tails like $1/\absv{\omega_j-\nu}^2$};
      \end{tikzpicture}\par
\caption{Schematic overlap of loop magnitudes away from their centers. The curves are illustrative, not frequency responses of the constructed blocks; the rigorous off-band estimates are given in Lemma~\ref{lem:tail}.}
\label{fig:tailoverlap}
\end{figure}
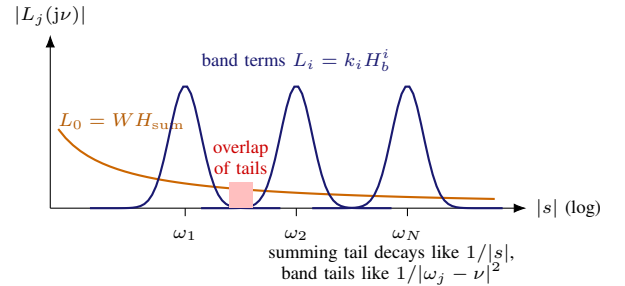
Near a crossing, a direct comparison likewise requires a quantitative bound relative to the small scalar return difference. Such a comparison may be possible when every scalar loop has sufficient margin, as will be the case for the gains constructed from an SSP solution. In the reverse direction, however, an arbitrary stabilizing gain may place some scalar factors at or near their crossing values. We therefore need a comparison that can determine the root count associated with one factor without requiring positive margins for all the others.

To this end, we divide the closed right half-plane into regions and compare the plant with \emph{one} scalar loop at a time. For a selected loop $i$, the perturbation of its return difference is precisely $\chi-(1+L_i)=\sum_{j\ne i}L_j$. We choose a region containing the selected factor's right half-plane roots and none of the other block denominators' roots, with a boundary on which the other loop terms are small. If this perturbation is smaller than the selected return difference on the boundary, Rouch\'e's theorem preserves the enclosed root count. Applying this argument separately in each region provides a rigorous way to use the familiar control-theoretic notion of frequency separation: each scalar loop determines the local root count when its margin is sufficient to dominate the interaction with the other loops. The comparison is made between $\Ptrue$ and $(1+L_i)D_{\mathrm{tot}}$, retaining all denominators so that the argument concerns polynomials and includes all internal modes. 

Throughout this section, assume $\max_i\abs{\ki_i}\le\kappa$ and $\abs\W\le w$, where $\kappa,w\ge1$. These bounds are \emph{hypotheses} of the comparison; suitable bounds for every stabilizing gain will be established in Section~\ref{sec:reduction}. Choosing the regions requires both root-location bounds and estimates on the loop terms away from their assigned frequencies. For the former, set
\begin{subequations}\label{eq:regionbounds}
    \begin{equation}\label{eq:regionbounds:roots}
\begin{aligned}
Q(\kappa)&\coloneqq1+S_d+\kappa(1+S_n),\\
R_W(w)&\coloneqq1+w+2\Wstar+(\Wstar)^2,
\end{aligned}
\end{equation}
which bound the moduli of the roots of $D_b+\ki N_b$ and $\Dsum+\W\Nsum$, respectively, under the stated gain bounds. Here $S_d,S_n$
are defined in Proposition~\ref{prop:Hb}\ref{propitem:coeff}. For the decay
estimates, set
 \begin{equation}\label{eq:regionbounds:decay}
R\coloneqq\max\{S_n,2S_d\}, \qquad R_H\coloneqq2\bigl(1+2\Wstar+(\Wstar)^2\bigr),
    \end{equation}
\end{subequations}
and Lemma~\ref{lem:decay} gives
\[
\begin{aligned}
\abs{\Hb(p)}&\le\frac4{\abs p^2}, &&\abs p\ge R,\\
\abs{\Hsum(s)}&\le\frac4{\abs s}, &&\abs s\ge R_H.
\end{aligned}
\]
For a band-pass copy, moving sufficiently far below or above its band in modulus makes $\abs{p_{\omega_i}(s)}$ large, so the prototype's decay controls the transformed block there. The following assumption ensures that these root and decay bounds apply on the comparison regions.
\begin{assumption}\label{ass:DC}
For given $\kappa\ge1$ and $w\ge1$, the center frequencies satisfy
\begin{enumerate}[label=(A\arabic*)]
\item\label{assitem:A1}
$\omega_{i+1}\ge4\omega_i$ for $1\le i\le N-1$;
\item\label{assitem:A2}
$\omega_1\ge4\max\{B_wQ(\kappa),R_W(w),R_H,B_wR\}$.
\end{enumerate}
\end{assumption}
The first condition separates consecutive bands, while the second places the first band above the summing dynamics and makes each band sufficiently narrow relative to its center frequency for the required localization. In particular, Lemma~\ref{lem:loc} places all roots of $\Dsum$ and $\Dsum+\W\Nsum$ in the disk $\abs s\le\omega_1/4$, and all roots of $D_i$ and $D_i+\ki_iN_i$ in the annulus
$2\omega_i/3\le\abs s\le3\omega_i/2$. We therefore define the larger comparison regions
\begin{equation}\label{eq:regions}
\begin{aligned}
\mathcal R_0&\coloneqq\{s\in\CRHP:\abs s\le\omega_1/2\},\\
\mathcal R_i&\coloneqq\{s\in\CRHP:\omega_i/2\le\abs s\le2\omega_i\},
\end{aligned}
\end{equation}
for $1\le i\le N$, and let
\[
\mathcal{G}\coloneqq\CRHP\setminus\bigcup_{i=0}^N\mathcal R_i.
\]
Each region contains all right half-plane roots of its corresponding scalar factor and none of the other block denominators, with a strict radial separation between the root-location sets and the circular boundaries. The regions have pairwise disjoint interiors, although adjacent regions may share a circular arc; see Fig.~\ref{fig:regions:1}.
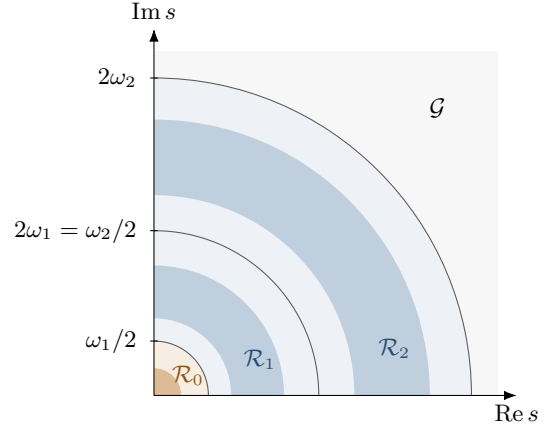
\begin{figure}[!hbt]
\centering
\begin{tikzpicture}[x=1cm,y=1cm,>=Latex,font=\small]
\definecolor{regionblue}{RGB}{44,96,145}
\definecolor{regionorange}{RGB}{180,105,24}
\fill[black!3] (0,0) rectangle (4.55,4.55);
\fill[regionblue!8] (0,0)--(4.2,0) arc[start angle=0,end angle=90,radius=4.2]--cycle;
\fill[regionorange!12] (0,0)--(.72,0) arc[start angle=0,end angle=90,radius=.72]--cycle;
\fill[regionblue!30] (2.65,0)--(3.65,0) arc[start angle=0,end angle=90,radius=3.65]--(0,2.65) arc[start angle=90,end angle=0,radius=2.65]--cycle;
\fill[regionblue!30] (1.02,0)--(1.72,0) arc[start angle=0,end angle=90,radius=1.72]--(0,1.02) arc[start angle=90,end angle=0,radius=1.02]--cycle;
\fill[regionorange!45] (0,0)--(.36,0) arc[start angle=0,end angle=90,radius=.36]--cycle;
\foreach \r in {.72,2.18,4.2}{\draw[black!65] (\r,0) arc[start angle=0,end angle=90,radius=\r];}
\draw[->] (0,0)--(4.8,0) node[below] {$\operatorname{Re}s$};
\draw[->] (0,0)--(0,4.85) node[above] {$\operatorname{Im}s$};
\foreach \r/\lab in {.72/{\omega_1/2},2.18/{2\omega_1=\omega_2/2},4.2/{2\omega_2}}{
 \draw (-.05,\r)--(.05,\r);\node[left=3pt] at (0,\r) {$\lab$};}
\node[regionorange!80!black] at (.45,.27) {$\mathcal R_0$};
\node[regionblue!80!black] at (1.40,.48) {$\mathcal R_1$};
\node[regionblue!80!black] at (3.18,.62) {$\mathcal R_2$};
\node at (3.75,3.8) {$\mathcal G$};
\end{tikzpicture}
\par
\caption{Comparison regions for $N=2$ and $\omega_2=4\omega_1$.
Darker sets contain the right-half-plane roots of the block denominators and scalar closed-loop factors: orange for the summing block and blue for the band blocks. The lower half follows by reflection across the real axis; radii are not to scale.}
\label{fig:regions:1}
\end{figure}
To control the interaction on these regions, define the common bound
\begin{equation}\label{eq:T}
    T\coloneqq \frac{8\abs\W}{\omega_1} +\frac{2(N+1)\kappa B_w^2}{\omega_1^2}.
\end{equation}
The first term bounds the summing-loop contribution outside its region, while the second bounds the combined off-band contributions of the band loops. Under Assumption~\ref{ass:DC}, Lemma~\ref{lem:tail} shows that $T$ bounds the sum of the other loop magnitudes on each $\partial\mathcal R_i$, the selected loop magnitude on its circular arcs, and the sum of all loop magnitudes in $\mathcal G$. Thus increasing the first frequency, while maintaining the separation, reduces all three bounds uniformly over the prescribed gain set.

The selected return difference is controlled differently on the two parts of the boundary. On the imaginary-axis segments its modulus is bounded below by the scalar modulus margin of Definition~\ref{def:margin}, which the band-pass transformation preserves. On the circular arcs, decay gives $\abs{1+L_i}\ge1-T$; see Fig.~\ref{fig:regions:2}. Consequently, a selected margin greater than $T$, together with $T\le1/4$, provides the domination needed for the local comparison. In $\mathcal G$, all loop terms are small, which will exclude additional right-half-plane roots. The following theorem states the resulting local counts and their simultaneous consequence; the supporting estimates are proved in Appendix~\ref{app:equiv}.
\begin{figure}[!hbt]
\centering
\begin{tikzpicture}[x=1cm,y=1cm,>=Latex,font=\small]
\definecolor{marginblue}{RGB}{37,91,148}
\definecolor{decayorange}{RGB}{173,87,17}
\fill[black!4] (0,-2)--(0,-.5) arc[start angle=-90,end angle=90,radius=.5]--(0,2) arc[start angle=90,end angle=-90,radius=2]--cycle;
\draw[->,black!40] (0,-2.3)--(0,2.4) node[above,black] {$\operatorname{Im}s$};
\draw[->,black!40] (-.16,0)--(2.45,0) node[right,black] {$\operatorname{Re}s$};
\draw[marginblue,line width=1.4pt] (0,.5)--(0,2);
\draw[marginblue,line width=1.4pt] (0,-2)--(0,-.5);
\draw[decayorange,line width=1.2pt,dashed] (0,-2) arc[start angle=-90,end angle=90,radius=2];
\draw[decayorange,line width=1.2pt,dashed] (0,-.5) arc[start angle=-90,end angle=90,radius=.5];
\node at (1.35,.30) {$\mathcal R_i$};
\node[left=3pt] at (0,2) {$2\omega_i$};
\node[left=3pt] at (0,.5) {$\omega_i/2$};
\node[left=3pt] at (0,-.5) {$-\omega_i/2$};
\node[left=3pt] at (0,-2) {$-2\omega_i$};
\node[anchor=west,align=left,text=marginblue] at (2.55,1.65)
 {Axis segments\\$|1+L_i|\ge m_b(\ki_i)$};
\draw[marginblue,->] (2.48,1.65)--(.07,1.45);
\node[anchor=west,align=left,text=decayorange] at (2.55,-1.05)
 {Circular arcs\\$|L_i|\le T$\\$|1+L_i|\ge1-T$};
\draw[decayorange,->] (2.48,-1.12)--(1.67,-1.08);
\end{tikzpicture}\par
\caption{Boundary estimates for a band region $\mathcal R_i$, $i\ge1$. The modulus margin controls the axis segments, while decay controls the circular arcs. On the entire boundary, the sum of the other loop magnitudes is at most $T$.}
\label{fig:regions:2}
\end{figure}
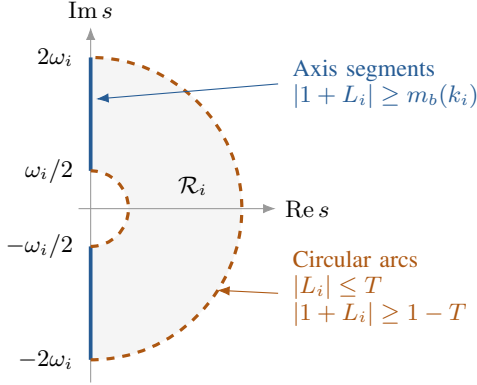
\begin{theorem}[Region-by-region root counts]\label{thm:localized}
Let Assumption~\ref{ass:DC} hold with $\kappa\ge\max_i\abs{\ki_i}$ and $w\ge\abs\W$, and assume $T\le1/4$. Set
\[
\begin{gathered}
Q_0\coloneqq\Dsum+\W\Nsum,\qquad
Q_i\coloneqq D_i+\ki_iN_i,\\
m_0\coloneqq m_s(\W),\qquad
m_i\coloneqq m_b(\ki_i),\quad 1\le i\le N.
\end{gathered}
\]
Then:
\begin{enumerate}[label=(\alph*)]
\item\label{thmitem:loc1}
$\Ptrue$ has no roots in $\mathcal G$.
\item\label{thmitem:loc2}
For every $i\in\{0,\dots,N\}$ with $m_i>T$,
\[
\Z(\Ptrue;\mathcal R_i)=\Z(Q_i;\CRHP).
\]
\item\label{thmitem:loc3} If $m_i>T$ for all $i$, then
\[
\Z(\Ptrue;\CRHP)=\Z(\Pnom;\CRHP).
\]
In particular, $\Ptrue$ is Hurwitz if and only if all $N+1$ factors of $\Pnom$ are Hurwitz.
\end{enumerate}
\end{theorem}
\begin{proof}
By Lemma~\ref{lem:loc}, every root of $Q_i$ in $\CRHP$ belongs to $\mathcal R_i$, and the other block denominators have no roots in that region. The same localization places all right-half-plane denominator roots away from the circular boundaries. Moreover, $\Dsum$ and the $D_i$ have no imaginary-axis roots by Proposition~\ref{prop:hsum}, Proposition~\ref{prop:Hb}\ref{propitem:nothurwitz}, and the band-pass transformation. Thus $D_{\mathrm{tot}}$ is nonzero on $\mathcal G$ and on every region boundary.

We first show how the tail estimates apply on these sets. By Lemma~\ref{lem:tail}\ref{lemitem:tail1} and \ref{lemitem:tail2},
\[
\begin{aligned}
\abs{L_0(s)}&\le\frac{8\abs\W}{\omega_1},
&&\abs s\ge\omega_1/2,\\
\abs{L_j(s)}&\le\frac{16\kappa B_w^2}{9\omega_j^2}
\le\frac{2\kappa B_w^2}{\omega_1^2},
&&\abs s\le\omega_j/2\text{ or }\abs s\ge2\omega_j.
\end{aligned}
\]
On a band region's boundary, condition~\ref{assitem:A1} places every other band term outside its own band, and the summing-loop estimate also applies. On $\partial\mathcal R_0$, only the band terms contribute to the perturbation, and all satisfy their off-band estimate. The selected loop itself satisfies its decay estimate on the circular arcs, while every loop satisfies the corresponding estimate in $\mathcal G$. Summing the relevant bounds gives the three conclusions of Lemma~\ref{lem:tail}\ref{lemitem:tail3}--\ref{lemitem:tail5}, each bounded by $T$. In particular, on $\mathcal G$,
\[
\abs{\chi-1}\le\sum_{j=0}^N\abs{L_j}\le T\le\frac14.
\]
Since $D_{\mathrm{tot}}$ is nonzero there, $\Ptrue=\chi D_{\mathrm{tot}}$ has no roots in $\mathcal G$, proving \ref{thmitem:loc1}.

For \ref{thmitem:loc2}, fix $i$ with $m_i>T$ and set
\[
P^{(i)}\coloneqq(1+L_i)D_{\mathrm{tot}}.
\]
This is the polynomial $Q_i$ multiplied by all other block denominators. On $\partial\mathcal R_i$, the perturbation bound gives
\[
\abs{\Ptrue-P^{(i)}} =\abs{D_{\mathrm{tot}}}\left|\sum_{j\ne i}L_j\right| \leq\abs{D_{\mathrm{tot}}}\,T.
\]
On the imaginary-axis segments, $\abs{1+L_i}\ge m_i>T$, using margin preservation under the band-pass transformation when $i\ge1$. On the circular arcs,
\[
\abs{1+L_i}\ge1-\abs{L_i}\ge1-T>T.
\]
Since $D_{\mathrm{tot}}$ is nonzero on the boundary, these inequalities give
\[
\abs{\Ptrue-P^{(i)}}<\abs{P^{(i)}} \qquad\text{on }\partial\mathcal R_i.
\]
Rouch\'e's theorem~\cite[Thm.~10.43(b)]{R:87} therefore yields
\[
\Z(\Ptrue;\mathcal R_i) =\Z(P^{(i)};\mathcal R_i) =\Z(Q_i;\CRHP).
\]
The last equality follows from localization: the other denominators contribute no roots in $\mathcal R_i$, and the region contains every root of $Q_i$ in $\CRHP$. The strict boundary inequality also excludes boundary roots, so the count on the closed region agrees with the interior count.

If every $m_i>T$, the local comparison applies to all regions. Together with \ref{thmitem:loc1}, it excludes roots in $\mathcal G$ and on every region boundary, so shared boundaries cause no double counting. Summing the local counts gives
\[
\Z(\Ptrue;\CRHP) =\sum_{i=0}^N\Z(Q_i;\CRHP) =\Z(\Pnom;\CRHP),
\]
which proves \ref{thmitem:loc3}.
\end{proof}

Part~\ref{thmitem:loc2} will be used in the reverse direction of the reduction, where it constrains a selected gain even if other gains lie near crossing values. Part~\ref{thmitem:loc3} will be used for the gains constructed from an SSP solution, whose scalar margins all exceed the coupling bound
\section{The reduction}\label{sec:reduction}
Theorem~\ref{thm:localized} provides the root-count comparison needed to relate stability of the coupled plant to the scalar constraints encoded by $\Pnom$. To apply it, we need bounds $\max_i\abs{\ki_i}\le\kappa$ and $\abs\W\le w$, and center frequencies satisfying Assumption~\ref{ass:DC} for the corresponding root and decay bounds $Q(\kappa)$, $R_W(w)$, $R$, and $R_H$. We must also ensure that the coupling bound $T$ is at most $1/4$ and strictly smaller than the scalar margins used in the comparison. Once the gain bounds and positive lower bounds on these margins are available, increasing $\omega_1$ and separating the subsequent frequencies makes the required comparisons possible.

Recall that the feedback gain itself must be unconstrained, yet these bounds must always hold to facilitate the proposed reduction. Therefore, we must establish such bounds for any stabilizing gain from properties of the constructed plant. We accomplish this by bounding $\W$ through necessary conditions for stability in terms of the coefficients, and then use the prototype's right half-plane zero to bound every gain entry. On this bounded gain set, we establish explicit positive lower bounds on the scalar margins away from the crossing gains. Note that excluding small neighborhoods of these crossings is necessary because the margins vanish there. These neighborhoods do not obstruct the reverse implication. If a gain lies near a crossing, it already belongs to an enlarged stability interval, while outside these neighborhoods its margin is sufficiently large for Theorem~\ref{thm:localized} to imply stability of the corresponding scalar factor whenever $\Ptrue$ is Hurwitz. Thus, in either case, the gain satisfies the enlarged interval constraint.

We then choose the interval widths, crossing neighborhoods, and center frequencies so that the required comparisons hold while the enlarged constraints still permit rounding by Lemma~\ref{lem:cssp}. Exact gains in $\{1,2\}$ will stabilize the plant for every YES instance, while any stabilizing gain will satisfy these enlarged constraints and hence yield a solution to the original SSP instance. Finally, we show that the resulting plant can be constructed in polynomial time and has polynomial bit length, completing the reduction.

\subsection{Enforcing bounds on stabilizing gains} \label{sec:gb}
We first bound $\W$ without restricting the individual gains. The band-loop terms in \eqref{eq:Ptrue} contribute neither to the coefficient of degree $1$ nor to that of degree $d-1$, so positivity of these two coefficients gives necessary bounds involving only $\W$.

\begin{proposition}[Necessary bounds on $\W$]\label{prop:coeff}
Let $d_4>0$ be the $p^4$ coefficient of $D_b$ and $q_0=(\Wstar)^2-\delta^2$, and
let $\gamma_1$ and $\gamma_{d-1}$ denote the coefficients of $\Ptrue(s,\K)$
of degrees $1$ and $d-1$.
Then
\[
\begin{gathered}
\begin{aligned}
   \gamma_1&=\Bigl(\prod_{j}\omega_j^{10}\Bigr)\Bigl[q_0\,d_4B_w\sum_{i}\omega_i^{-2}+2\Wstar-\W\Bigr]\\
\gamma_{d-1}&=\W+N\,d_4B_w . 
\end{aligned}
\end{gathered}
\]
If $\Ptrue(s,\K)$ is Hurwitz, then
\[
-N\,d_4B_w\;<\;\W\;<\;2\Wstar+q_0\,d_4B_w\sum_{i}\omega_i^{-2}.
\]
\end{proposition}
\GB{
\begin{proof}
Each term $\ki_iN_i\Dsum\prod_{j\ne i}D_j$ in \eqref{eq:Ptrue} is divisible by $s^2$ and has degree at most $d-2$, so it contributes to neither coefficient. By \eqref{eq:NiDi}, $D_i(0)=\omega_i^{10}$, the coefficient of $s$ in $D_i$ is $d_4B_w\omega_i^8$, and its subleading coefficient is $d_4B_w$. Since $\Dsum=s^3+2\Wstar s+q_0$ and $\Nsum=s^2-s$, the degree-one coefficient of $(\Dsum+\W\Nsum)\prod_jD_j$ is
\[
\Bigl(\prod_j\omega_j^{10}\Bigr) \Bigl[q_0d_4B_w\sum_i\omega_i^{-2}+2\Wstar-\W\Bigr],
\]
while its $d-1$ coefficient is $\W+Nd_4B_w$. A monic real Hurwitz polynomial must have positive coefficients, so these identities give the stated strict bounds.
\end{proof}
The frequency conditions simplify the upper bound independently of the individual gains. Under Assumption~\ref{ass:DC}, $\sum_i\omega_i^{-2}\le2/\omega_1^2$ by \ref{assitem:A1}. Also, \ref{assitem:A2} gives $\omega_1\ge4R_H=8(1+\Wstar)^2$ and $\omega_1\ge4B_wR\ge56B_w$, since $R\ge2S_d\ge2d_3\ge14$. Thus
\begin{equation}\label{eq:omega1sq}
\omega_1^2\ge448(1+\Wstar)^2B_w >\frac{12}{7}(\Wstar)^2B_w \ge2q_0d_4B_w,
\end{equation}
where $d_4\le6/7$ and $q_0<(\Wstar)^2$. Consequently every stabilizing gain satisfies $-Nd_4B_w<\W<2\Wstar+1$, and hence
\[
\abs\W\le\bar w\coloneqq2\Wstar+1+Nd_4B_w.
\]
This controls the summing-loop contribution in the following argument, even if individual gain entries are large.
}

Next, we use the prototype's right half-plane zero to bound every stabilizing gain entry. Near the corresponding zero of a band copy, the term with largest absolute gain dominates the other loop terms when the frequencies are sufficiently separated. A polynomial comparison then forces a right half-plane root of $\Ptrue$.

Set
\[
\mu\coloneqq\frac1{500},\qquad \kappa^*\coloneqq\frac2\mu=1000,\qquad
T_{\mathrm{GB}}\coloneqq \frac{8\bar w}{\omega_1}+\frac{2NB_w^2}{\omega_1^2},
\]
where $\mu$ is the lower bound in Proposition~\ref{prop:Hb}\ref{propitem:noroots}.

\begin{proposition}[A priori gain bound]\label{prop:GB}
Let the frequencies satisfy Assumption~\ref{ass:DC} for $(\kappa,w)=(1,\bar w)$, 
and suppose that $T_{\mathrm{GB}} \leq  \frac{\mu}{4}$. If $\Ptrue(s,\K)$ is Hurwitz, then $\max_i\abs{\ki_i}\le\kappa^*$.
\end{proposition}
\begin{proof}
Suppose that $\Ptrue(s,\K)$ is Hurwitz and $M\coloneqq\max_i\abs{\ki_i}>\kappa^*$, and let $i^*$ attain the maximum. By Proposition~\ref{prop:coeff} and the preceding consequence of Assumption~\ref{ass:DC} we have $\abs\W\le\bar w$. Now, compare with $\ki_{i^*}\Hb^{i^*}$ near its right-half-plane zero. The other band gains have modulus at most $M$, while the summing gain has the independent bound $\bar w$. Only the open-loop denominators are localized in this argument.

Write $\omega\coloneqq\omega_{i^*}$, and let $z'$ be the root with positive imaginary part of $s^2-3B_ws+\omega^2=0$. Then $\abs{z'}=\omega$, $\operatorname{Re}z'=3B_w/2$, and $p_\omega(z')=3$. Since $\omega\ge\omega_1\ge4B_wR\ge56B_w$, the disk
\[
\Delta\coloneqq\{s:\abs{s-z'}\le3B_w/8\}
\]
lies in the half-plane $\operatorname{Re}s\ge9B_w/8$ and annulus $\omega/2<\abs s<2\omega$. Now, from \eqref{eq:pw} we have for $s\ne0$,
\[
p_\omega(s)-3=\frac{(s-z')(s-\bar z')}{B_ws}.
\]
Using $\omega-9B_w^2/(4\omega)\le\operatorname{Im}z'\le\omega$
gives
\[
\begin{aligned}
\abs{p_\omega(s)-3}
&\le\frac38\frac{2\omega+3B_w/8}{\omega-3B_w/8}\le0.76,\\
&\hspace{3em}s\in\Delta,\\
\abs{p_\omega(s)-3}
&\ge\frac38\frac{2\omega-9B_w^2/(2\omega)-3B_w/8}{\omega+3B_w/8}\\
&\ge0.74,\qquad s\in\partial\Delta.
\end{aligned}
\]
Note that the upper expression decreases and the lower increases with $\omega/B_w\ge56$; and that their endpoint values are $2697/3560<0.76$ and $37479/50512>0.74$ respectively. Thus $p_\omega(\Delta)$ lies in $\abs{p-3}<7/8$, and its boundary image lies in $5/8\le\abs{p-3}\le7/8$. Proposition~\ref{prop:Hb}\ref{propitem:noroots} therefore excludes zeros of $D_{i^*}$ in $\Delta$ and gives $\abs{\Hb^{i^*}}\ge\mu$ on $\partial\Delta$.

Lemma~\ref{lem:loc}, applied to the open-loop denominators, and \ref{assitem:A1} place the roots of every other denominator at modulus at most $3\omega/8$ or at least $8\omega/3$. Hence $D_{\mathrm{tot}}$ has no zeros in $\Delta$. By Proposition~\ref{prop:Hb}\ref{propitem:hurwitz}, the only possible numerator zero there satisfies $p_\omega(s)=3$. Of its two preimages $z',\bar z'$, only $z'$ belongs to $\Delta$, and it is simple since $p_\omega'(z')=(z'-\bar z')/(B_wz')\ne0$. Consequently $D_{\mathrm{tot}}\Hb^{i^*}=N_{i^*}\Dsum\prod_{j\ne i^*}D_j$ is a polynomial with exactly
one zero in $\Delta$.

On $\partial\Delta$, every other band is off-band by \ref{assitem:A1}, and $\abs s>\omega_1/2$. Lemma~\ref{lem:tail}\ref{lemitem:tail1} and \ref{lemitem:tail2} give $\abs{\Hsum}\le8/\omega_1$ and $\abs{\Hb^j}\le2B_w^2/\omega_1^2$ for $j\ne i^*$. Using \eqref{eq:chitrue}, we obtain
\[
\begin{aligned}
&\abs{\Ptrue-\ki_{i^*}D_{\mathrm{tot}}\Hb^{i^*}}\\
&\quad\le\abs{D_{\mathrm{tot}}}
\left(1+\frac{8\bar w}{\omega_1}+\frac{2(N-1)MB_w^2}{\omega_1^2}\right)\\
&\quad\le\abs{D_{\mathrm{tot}}}(1+T_{\mathrm{GB}}+MT_{\mathrm{GB}})\\
&\quad<M\mu\abs{D_{\mathrm{tot}}}
\le\abs{\ki_{i^*}D_{\mathrm{tot}}\Hb^{i^*}}.
\end{aligned}
\]
The strict inequality follows from $T_{\mathrm{GB}}\le\mu/4$, $M\ge1$, and $M>2/\mu$, since $1+T_{\mathrm{GB}}+MT_{\mathrm{GB}} \le1+M\mu/2<M\mu$. Both compared expressions are polynomials, and $D_{\mathrm{tot}}\ne0$ on the boundary. Rouch\'e's theorem~\cite[Thm.~10.43(b)]{R:87} gives $\Z(\Ptrue;\Delta)=1$, contradicting the Hurwitzness.
\end{proof}

\subsection{Explicit modulus margin bounds}\label{sec:explicit}
The preceding bounds restrict every stabilizing gain to ranges on which we can estimate the scalar modulus margins uniformly. To make the coupling smaller than these margins using frequencies of polynomial bit length, we need \emph{computable} lower bounds rather than qualitative positivity. The following bounds depend explicitly on the distance from the crossing gains, and Definition~\ref{def:schedule} will use them at the exact gains of the forward direction and outside the crossing neighborhoods of the reverse direction.

\begin{lemma}[\rev{Explicit margin bounds}]\label{lem:explicit}
\rev{Let $\eps\in(0,1/8]$, $\Wstar\in\mathbb Z_+$ and $\delta\in(0,1/4]$ be the CSSP parameters. }
\begin{enumerate}[label=(\alph*)]
\item\label{lemitem:explicit1}\rev{Let $\kappa\ge1$ and $R=\max\{S_n,2S_d\}$ as in Section~\ref{sec:equiv}. Set
\[
\begin{gathered}
C_S=14+7\kappa,\qquad Y=\max\{R^2,8\kappa\},\\
M_b(t)=\frac{t^4}{114C_S^4Y^4},\qquad 0<t\le1.
\end{gathered}
\]
For every real $g$ with $\abs{g}\le\kappa$ and distance $\tau>0$ from $\{1-\eps,1+\eps,2-\eps,2+\eps\}$, set $\hat\tau=\min\{\tau,1\}$. Then $m_b(g)\ge M_b(\hat\tau)$.}
\item\label{lemitem:explicit2}\rev{Let $w\ge1$ and set
\[
\begin{gathered}
R_H=2\bigl(1+2\Wstar+(\Wstar)^2\bigr),\\
Y_s=\max\{R_H^2,64w^2\},\\
D_s=\bigl(1+2\Wstar+(\Wstar)^2\bigr)Y_s^2,\\
M_s(t)=\frac{15t^2}{32(w+1)^2D_s},\qquad 0<t\le1.
\end{gathered}
\]
For every real $\W$ with $\abs{\W}\le w$ and distance $\tau_W>0$ from $\{\Wstar-\delta,\Wstar+\delta\}$, set $\hat\tau_W=\min\{\tau_W,1\}$. Then $m_s(\W)\ge M_s(\hat\tau_W)$.}
\end{enumerate}
\end{lemma}

\begin{proof}
\rev{For~\ref{lemitem:explicit1}, write $P_g=D_b+gN_b$. As in Lemma~\ref{lem:cross}, with $y=\nu^2\ge0$,
\[
\begin{gathered}
P_g(\J\nu)=f(y)+\J\nu h(y),\\
f(y)=c_4y^2-c_2y+c_0,\qquad h(y)=y^2-c_3y+c_1,\\
|P_g(\J\nu)|^2=f(y)^2+yh(y)^2.
\end{gathered}
\]
Since $D_b$ has no imaginary-axis zeros we can divide by $\abs{D_b(\J\nu)}$ and obtain the return difference. For $y>0$, its numerator can be small only if both $f(y)$ and $h(y)$ are small. We first rule out this simultaneous behavior quantitatively, then treat the origin separately because the weight of $h(y)^2$ vanishes there.}

\rev{The coefficient bounds~\eqref{eq:coeffbounds} give
\[
\begin{gathered}
\abs{c_4}\le1,\quad \abs{c_3}\le9+\kappa,\quad \abs{c_2}\le2+3\kappa,\\
\abs{c_1},\abs{c_0}\le4+4\kappa.
\end{gathered}
\]
Thus $C_S$ bounds both $\abs{c_4}+\abs{c_2}+\abs{c_0}$ and $1+\abs{c_3}+\abs{c_1}$. Also, since $\lambda\ge5/4$,
\[
\abs{c_0}=\frac{12\lambda}{7}\abs{g-(1-\eps)}
\ge\frac{15}{7}\tau\ge\hat\tau.
\]
The other three crossing gains appear in the resultant computed in Lemma~\ref{lem:cross}:
\[
\begin{aligned}
\abs{\Res_y(f,h)}
&=\frac{8(4\eps+5)^2(2-11\eps)^2}{49(1-4\eps)^3}\\
&\quad\cdot|(\eps-g+1)(\eps-g+2)(\eps+g-2)|\\
&\ge\frac{625}{392}\tau^3\ge\hat\tau^3.
\end{aligned}
\]
Here each linear factor has modulus at least $\tau$, while $4\eps+5\ge5$, $2-11\eps\ge5/8$ and $0<1-4\eps\le1$.}

\rev{As the resultant does not vanish there is no common zero. Nevertheless, a lower bound on the return difference requires a bounded polynomial identity. Since $c_4=d_4>0$, both $f$ and $h$ have degree two. Let $\mathcal S$ be their $4\times4$ Sylvester matrix with rows given by $yf,f,yh,h$ in the basis $y^3,y^2,y,1$. If $(u_1,u_0,v_1,v_0)$ is the last row of $\operatorname{adj}\mathcal S$, define the degree-at-most-one B\'ezout polynomials
\[
U(y)=u_1y+u_0,\qquad V(y)=v_1y+v_0.
\]
The identity $(\operatorname{adj}\mathcal S)\mathcal S=(\det\mathcal S)I$ gives
\[
U(y)f(y)+V(y)h(y)=\Res_y(f,h).
\]
Each coefficient of $U,V$ is a signed $3\times3$ minor. Since each row norm of $\mathcal S$ is at most $C_S$, Hadamard's inequality bounds these coefficients by $C_S^3$. Consequently, for $y\ge0$,
\[
\begin{aligned}
|U(y)|+|V(y)|&\le4C_S^3\max\{1,y\},\\
\max\{|f(y)|,|h(y)|\}
&\ge\frac{\hat\tau^3}{4C_S^3\max\{1,y\}}.
\end{aligned}
\]
This estimate controls intermediate frequencies. Near zero we use $c_0$, and beyond $Y$ the decay of $\Hb$ suffices. Set $y_0=\hat\tau/(2C_S)\le1$.}
\begin{enumerate}[label=(a.\arabic*)]
\item\rev{If $y\ge Y$, then $\nu\ge R$ for $\nu\ge0$ and Lemma~\ref{lem:decay} gives $|g\Hb(\J\nu)|\le4\kappa/y\le1/2$. Hence $|1+g\Hb(\J\nu)|\ge1/2$.}
\item\label{proofitem:explicit_range2}\rev{If $0\le y\le y_0$, then $|f(y)|\ge|c_0|-(|c_2|+|c_4|)y\ge\hat\tau/2$.}
\item\label{proofitem:explicit_range3}\rev{If $y_0\le y\le Y$, then
\[
\begin{aligned}
f^2+yh^2&\ge\min\{1,y\}(f^2+h^2)\\
&\ge y_0\max\{f^2,h^2\}
\ge\frac{\hat\tau^7}{32C_S^7Y^2}.
\end{aligned}
\]}
\end{enumerate}
\rev{In ranges~\ref{proofitem:explicit_range2} and~\ref{proofitem:explicit_range3}, $S_d\le19$ gives $|D_b(\J\nu)|\le(1+S_d)\max\{1,|\nu|\}^5\le20Y^{5/2}$. The moduli are even in $\nu$, so these three ranges cover all real frequencies and yield
\[
m_b(g)^2\ge\min\left\{\frac14,\frac{\hat\tau^2}{1600Y^5},
\frac{\hat\tau^7}{12800C_S^7Y^7}\right\}.
\]
Since $\hat\tau\le1$ and $C_S,Y\ge1$, the third entry is the smallest. Using $12800\le114^2$ bounds it below by $M_b(\hat\tau)^2$, proving~\ref{lemitem:explicit1}.}

\rev{For~\ref{lemitem:explicit2}, put $q_0=(\Wstar)^2-\delta^2\ge15/16$. The summing numerator at $\J\nu$ has the same form $f(y)+\J\nu h(y)$, now with
\[
f(y)=q_0-\W y,\qquad h(y)=2\Wstar-\W-y.
\]
Here elimination is immediate: $f-\W h=(\W-\Wstar)^2-\delta^2$. The two factors on the right have modulus at least $\tau_W$, so for every $y\ge0$,
\[
\max\{|f(y)|,|h(y)|\}\ge\frac{\hat\tau_W^2}{1+w}.
\]
We again combine a constant-term estimate near zero, this estimate at intermediate frequencies, and decay at high frequencies. Set $y_W=15/(32(w+1))\le1$.}
\begin{enumerate}[label=(b.\arabic*)]
\item\rev{If $y\ge Y_s$, then $|\nu|\ge R_H$ and $|\nu|\ge8w$, so Lemma~\ref{lem:decay} gives $|\W\Hsum(\J\nu)|\le4w/|\nu|\le1/2$.}
\item\label{proofitem:explicit_srange2}\rev{If $0\le y\le y_W$, then $|f(y)|\ge q_0-wy_W\ge q_0-15/32\ge q_0/2$.}
\item\label{proofitem:explicit_srange3}\rev{If $y_W\le y\le Y_s$, then $f^2+yh^2\ge y_W\max\{f^2,h^2\}\ge y_W\hat\tau_W^4/(1+w)^2$.}
\end{enumerate}
\rev{In the last two ranges, $|\Dsum(\J\nu)|\le(1+2\Wstar+q_0)\max\{1,|\nu|\}^3\le D_s$. Since $\Dsum$ has no imaginary-axis zeros by Proposition~\ref{prop:hsum}, division is valid and
\[
m_s(\W)\ge\min\left\{\frac12,\frac{q_0}{2D_s},
\frac{\sqrt{y_W}\hat\tau_W^2}{(1+w)D_s}\right\}\ge M_s(\hat\tau_W).
\]
The last inequality follows from $\sqrt{y_W}\ge y_W$, $\hat\tau_W\le1$, $q_0\ge15/16$ and $D_s\ge1$.}
\end{proof}

\subsection{Parameter choice}

The following definition fixes all parameters of the plant as functions of the
SSP instance.

\begin{definition}[Parameter choice]\label{def:schedule}
Given an SSP instance $(a,S)$ with $\kappa^*=1000$ as in
Section~\ref{sec:gb}, set
\[
\begin{gathered}
\Wstar\coloneqq S+\sum_ia_i,\\
\delta\coloneqq \tfrac14,\\
\eps\coloneqq \frac{1}{8\sum_ia_i},\\
\theta\coloneqq \frac{\eps}{2},\\
\theta_W\coloneqq \tfrac1{16},\\
B_w\coloneqq 1 .
\end{gathered}
\]
Moreover, set
\[
\begin{gathered}
X_b^\theta\coloneqq \bigcup_{c\in\{1-\eps,\,1+\eps,\,2-\eps,\,2+\eps\}}(c-\theta,\,c+\theta),
\\
X_s^{\theta_W}\coloneqq \bigcup_{c\in\{\Wstar-\delta,\,\Wstar+\delta\}}(c-\theta_W,\,c+\theta_W),
\end{gathered}
\]
the unions of the open $\theta$-neighborhoods of the crossing gains and of the
open $\theta_W$-neighborhoods of $\Wstar\pm\delta$,
\[
\bar m\coloneqq \min\bigl\{M_b(\theta),\ M_s(\theta_W)\bigr\},
\]
with $M_b$ from Lemma~\ref{lem:explicit}(a) for $\kappa=\kappa^*$ and $M_s$ from
Lemma~\ref{lem:explicit}(b) for $w=\bar w$, and choose $\omega_1$ as the
smallest power of two strictly greater than
\begin{align*}
\max\Bigl\{&4Q(\kappa^*),\ 4R_W(\bar w),\ 4R_H,\ 4R,\\
&2000\bigl(8\bar w+2N\bigr),\\
&\bigl(8\bar w+2(N+1)\kappa^*\bigr)\max\{4,\ 1/\bar m\}\Bigr\},
\end{align*}
and $\omega_i\coloneqq 4^{\,i-1}\omega_1$.
\end{definition}
The plant obtained from the parameter choice above is the realization of
Lemma~\ref{lem:realization} with these parameters.

The next lemma states the properties of the parameter choice that the
correctness proofs use.
\begin{lemma}[Properties of the parameter choice]\label{lem:schedule}
The parameter choice of Definition~\ref{def:schedule} has the following
properties.
\begin{enumerate}[label=(\alph*)]
\item Assumption~\ref{ass:DC} holds for every pair of bounds
$\kappa\in[1,\kappa^*]$ and $w\in[1,\bar w]$, and the hypotheses of
Proposition~\ref{prop:GB} hold. \label{lemitem:schedule_ass}
\item $m_b(1)\ge\bar m$, $m_b(2)\ge\bar m$, and $m_s(\Wstar)\ge\bar
    m$.\label{lemitem:schedule_bounds1}
\item If $\abs g\le\kappa^*$ and $g\notin X_b^{\theta}$, then $m_b(g)\ge\bar m$. Additionally,
if $\abs{\W}\le\bar w$ and $\W\notin X_s^{\theta_W}$, then $m_s(\W)\ge\bar m$.
\label{lemitem:schedule_bounds2}
\item For every $\K$ with $\max_i\abs{\ki_i}\le\kappa^*$ and
$\abs{\W}\le\bar w$, the coupling bound $T$ for
$(\kappa,w)=(\kappa^*,\bar w)$ satisfies $T<\min\{\bar m,\,1/4\}$.
\label{lemitem:schedule_K}
\end{enumerate}
\end{lemma}
\begin{proof}
    \ref{lemitem:schedule_ass} \ref{assitem:A1} holds since $\omega_{i+1}=4\omega_i$.
The functions $Q$ and $R_W$ are increasing, so \ref{assitem:A2} for the bounds
$(\kappa^*,\bar w)$ implies \ref{assitem:A2} for all smaller bounds.
With $B_w=1$, the frequency $\omega_1$ exceeds $4Q(\kappa^*)$, $4R_W(\bar w)$,
$4R_H$ and $4R=4B_wR$ by construction.
For Proposition~\ref{prop:GB}: 
$\omega_1>2000\bigl(8\bar w+2N\bigr)$ together with $\omega_1\ge1$
gives
\[
T_{\mathrm{GB}}\ \le\ \frac{8\bar w+2N}{\omega_1}\ <\ \frac{1}{2000}\ =\ \frac{\mu}{4}.
\]

\ref{lemitem:schedule_bounds1}, \ref{lemitem:schedule_bounds2} Lemma~\ref{lem:explicit} applies with $\kappa=\kappa^*$ and $w=\bar w$,
since $\eps=1/(8\sum_ia_i)\le1/8$ and $\delta=1/4$.
The distance from $1$ and from $2$ to the crossing gains $\{1\pm\eps,2\pm\eps\}$
is $\eps$, and the distance from $\Wstar$ to $\{\Wstar\pm\delta\}$ is $\delta$;
both are at most $1$.
A gain at distance at least $\theta$ from the crossing gains has
$\hat\tau\ge\theta$, because $\theta\le1$.
Likewise $\hat\tau_W\ge\theta_W$ at distance at least $\theta_W$.
The functions $M_b$ and $M_s$ are increasing on $(0,1]$, since $M_b(t)$ is a
positive multiple of $t^4$ and $M_s(t)$ of $t^2$, and $\theta\le\eps$,
$\theta_W\le\delta$.
The lemma therefore gives $m_b(1),m_b(2)\ge M_b(\eps)\ge M_b(\theta)\ge\bar m$
and $m_s(\Wstar)\ge M_s(\delta)\ge M_s(\theta_W)\ge\bar m$.
The same monotonicity gives both implications of \ref{lemitem:schedule_bounds2}.

\ref{lemitem:schedule_K} With these bounds,
$T\le8\bar w/\omega_1+2(N+1)\kappa^*/\omega_1^{2}\le\bigl(8\bar w+2(N+1)\kappa^*\bigr)/\omega_1$,
using $\omega_1\ge1$.
By construction
$\omega_1>\bigl(8\bar w+2(N+1)\kappa^*\bigr)\max\{4,\,1/\bar m\}$, so
$T<\min\{\bar m,\,1/4\}$.
\end{proof}

\subsection{Correctness}

The next proposition proves the forward direction of the reduction.

\begin{proposition}[Sufficiency]\label{prop:sufficiency}
Assume the parameter choice of Definition~\ref{def:schedule}.
If the SSP instance is a YES instance, then there exists
$\hat\K\in\{1,2\}^{1\times N}$ with $\sum_ia_i\hat\ki_i=\Wstar$.
For every such $\hat\K$ the polynomial $\Ptrue(s,\hat\K)$ is Hurwitz, so the
plant is stabilizable.
\end{proposition}

\begin{proof}
If $\sum_ia_ix_i=S$ with $x\in\{0,1\}^N$, then the row $\hat\K$ with entries
$\hat\ki_i\coloneqq x_i+1$ satisfies $\sum_ia_i\hat\ki_i=S+\sum_ia_i=\Wstar$.
Let $\hat\K\in\{1,2\}^{1\times N}$ satisfy $\sum_ia_i\hat\ki_i=\Wstar$.
At $\K=\hat\K$ the bounds $\max_i\abs{\hat\ki_i}\le2\le\kappa^*$ and
$\abs\W=\Wstar\le\bar w$ hold.
The margins are $m_i\in\{m_b(1),m_b(2)\}$ and $m_0=m_s(\Wstar)$, and each is
at least $\bar m$ by Lemma~\ref{lem:schedule}(b).
Lemma~\ref{lem:schedule}(a) gives Assumption~\ref{ass:DC} with bounds
$(\kappa^*,\bar w)$, and Lemma~\ref{lem:schedule}(d) gives
$T<\min\{\bar m,1/4\}$.
Theorem~\ref{thm:localized}(c) applies, and every nominal factor is Hurwitz.
The factors $D_b+1\cdot N_b$ and $D_b+2\cdot N_b$ are Hurwitz by
Proposition~\ref{prop:Hb}\ref{propitem:hurwitz2}, so $D_i+\hat\ki_iN_i$ is
Hurwitz for every $i$ by Proposition~\ref{prop:bpmap}\ref{propitem:bpmap2}, and
$\Dsum+\Wstar\Nsum$ is Hurwitz by Proposition~\ref{prop:hsum}.
Hence $\Z(\Ptrue;\CRHP)=0$.
\end{proof}

The next proposition proves the reverse direction.

\begin{proposition}[Necessity]\label{prop:necessity}
Assume the parameter choice of Definition~\ref{def:schedule}.
If some $\K\in\mathbb{R}^{1\times N}$ makes $\Ptrue(s,\K)$ Hurwitz, then the
SSP instance is a YES instance.
\end{proposition}

\begin{proof}
The hypotheses of Proposition~\ref{prop:GB} hold by Lemma~\ref{lem:schedule}(a),
so $\max_i\abs{\ki_i}\le\kappa^*$.
By Proposition~\ref{prop:coeff} and Assumption~\ref{ass:DC}, which holds by
Lemma~\ref{lem:schedule}(a), we get $-N\,d_4<\W<2\Wstar+1$, hence
$\abs{\W}\le\bar w$.
By Lemma~\ref{lem:schedule}(a) and (d), Assumption~\ref{ass:DC} with bounds
$(\kappa^*,\bar w)$ and the inequality $T<\min\{\bar m,1/4\}$ hold at $\K$.

Fix $i\in\{1,\dots,N\}$.
Either $\ki_i\in X_b^\theta$, and then $\ki_i$ lies within $\theta$ of one of
$1\pm\eps,2\pm\eps$, so
$\ki_i\in(1-\eps',1+\eps')\cup(2-\eps',2+\eps')$ with $\eps'\coloneqq \eps+\theta$.
Or $\ki_i\notin X_b^\theta$.
Then $\abs{\ki_i}\le\kappa^*$, so
$m_b(\ki_i)\ge\bar m>T$ by Lemma~\ref{lem:schedule}(c), and
Theorem~\ref{thm:localized}(b) for the region $\mathcal{R}_{i}$ gives
\[
\Z\bigl(D_i+\ki_iN_i;\CRHP\bigr)=\Z\bigl(\Ptrue;\mathcal{R}_{i}\bigr)=0,
\]
because $\Ptrue$ is Hurwitz.
So $D_i+\ki_iN_i$ is Hurwitz, and $\ki_i\in(1-\eps,1+\eps)\cup(2-\eps,2+\eps)$
by Propositions~\ref{prop:bpmap}\ref{propitem:bpmap2} and~\ref{prop:Hb}\ref{propitem:hurwitz2}.
In particular $\ki_i$ lies in the enlarged intervals.
The same dichotomy for the summing factor and the region $\mathcal{R}_0$, with
$\abs\W\le\bar w$, the set $X_s^{\theta_W}$, and Lemma~\ref{lem:schedule}(c),
gives $\abs{\W-\Wstar}<\delta'\coloneqq \delta+\theta_W$.

The enlarged parameters satisfy
\[
\delta'+\eps'\sum_ia_i=\tfrac14+\tfrac1{16}+\tfrac{3}{16}=\tfrac12<1 .
\]
Hence $\K$ witnesses that $(a,\Wstar,\eps',\delta')$ is a YES instance of
CSSP, and Lemma~\ref{lem:cssp} gives $\hat\K\in\{1,2\}^{1\times N}$ with
$\sum_ia_i\hat\ki_i=\Wstar$, that is $\sum_ia_ix_i=S$ for $x_i=\hat\ki_i-1$.
\end{proof}

The two directions combine into the reduction.

\begin{theorem}[Reduction]\label{thm:main}
The map from SSP instances  $(a,S)$ to the plant $(A,B,C)$ of
Definition~\ref{def:schedule} is computable in polynomial time, and an SSP instance is a YES instance
if and only if its image is a YES instance of unconstrained SOFS.
\end{theorem}

\begin{proof}
The equivalence is Propositions~\ref{prop:sufficiency} and~\ref{prop:necessity},
together with Lemma~\ref{lem:realization}.
For the size, Lemma~\ref{lem:realization} gives that the bit length of $(A,B,C)$
is polynomial in the bit lengths of $\eps$, $\delta$, $\Wstar$, the $a_i$,
$B_w=1$, and the $\omega_i$.
The first four have bit length polynomial in the instance size.
For $\omega_i=4^{i-1}\omega_1$ it suffices that $\log_2\omega_1$ is polynomial
in the instance size.
This holds because $\log_2\bar w$ is polynomial in the instance size.
So is $\log_2(1/\bar m)$, by the explicit formulas of Lemma~\ref{lem:explicit}
with $\theta=1/(16\sum_ia_i)$, $\theta_W=1/16$, $\kappa=1000$ and $w=\bar w$. \rev{The rational bounds in Definition~\ref{def:schedule} are obtained by rational arithmetic and comparisons on numbers of polynomial bit length. The smallest power of two strictly above their maximum is found from its binary length and an exact comparison, so it too is computable in polynomial time and has polynomial bit length.}
Hence the parameter choice is computable in polynomial time.
\end{proof}

The main result follows.
\begin{theorem}[NP-hardness]\label{thm:nphard}
Unconstrained SOFS is NP-hard.
\end{theorem}
\begin{proof}
SSP is NP-complete, hence NP-hard, and the map of Theorem~\ref{thm:main} is a
Karp reduction from SSP to unconstrained SOFS. By composition of Karp reductions, this proves the claim.
\end{proof}

\section{Concluding remarks}\label{sec:concl}
Theorem~\ref{thm:nphard} states that unconstrained static output feedback
stabilization is NP-hard. The reduction of Theorem~\ref{thm:main} maps a Subset Sum instance with $N$ items to a plant with one input, $N$ outputs, and $10N+3$ states. Hardness transfers to plants with $N$ inputs and one output. The matrices $A+B\K C$ and $A^{\mathsf T}+C^{\mathsf T}\K^{\mathsf T}B^{\mathsf T}$ have the same eigenvalues, and $\K\mapsto\K^{\mathsf T}$ is a bijection of the gain sets. Hence $(A,B,C)$ is a YES instance of SOFS if and only if $(A^{\mathsf T},C^{\mathsf T},B^{\mathsf T})$ is. Composing the map of Theorem~\ref{thm:main} with the transposition therefore gives a Karp reduction whose images have $N$ inputs and one output.

We conclude this paper by making a straightforward connection with another
complexity class, namely PSPACE. First, we have the following definition. 
\begin{definition}[{PSPACE \cite[Defs.~4.1 and~4.5]{AB:09}}]\label{def:pspace}
A decision problem $\Pi$ is in PSPACE if there exist a polynomial $\pi$ and a
Turing machine that decides $\Pi$ and, on every instance $\mathcal{I}$, uses at
most $\pi(\len(\mathcal{I}))$ cells of its work tapes.
\end{definition}
Unconstrained SOFS lies in PSPACE.
By Lyapunov's theorem, $A+B\K C$ is Hurwitz stable if and only if there exists a
symmetric $P\succ0$ with $(A+B\K C)^{\mathsf T}P+P(A+B\K C)\prec0$.
A symmetric matrix is positive definite if and only if it equals
$LL^{\mathsf T}$ for a lower-triangular $L$ with positive diagonal entries.
An instance is therefore equivalent to an existential sentence over the reals.
Its (quantified) variables are the entries of $\K$, of $P$, and of two lower-triangular
factors.
Its constraints are polynomial equations and strict inequalities of degree at
most two, with coefficients rational in the data, and its length is polynomial
in the instance size.
Deciding such sentences is in PSPACE~\cite{Can:88}.

Three questions remain open.
The first is hardness in the strong sense.
The constructed plant contains rationals whose magnitudes are exponential in the
instance size.
Two features of the parameter choice cause this.
The frequency $\omega_1$ exceeds a constant multiple of $(\sum_ia_i)^4$ by
Definition~\ref{def:schedule}, since $1/\bar m\ge1/M_b(\theta)$ with
$\theta=1/(16\sum_ia_i)$, so it is exponential in the bit length of the $a_i$.
The spacing $\omega_i=4^{i-1}\omega_1$ makes $\omega_N$ exponential in $N$.
A hardness proof in the strong sense would also need a source problem that is
strongly NP-hard, which Subset Sum is not.
It is not known whether unconstrained SOFS remains NP-hard when the magnitudes
of all numerators and denominators of the plant entries are bounded by a
polynomial in the instance size.
The second is membership in NP, which would follow if every stabilizable
rational instance admitted a stabilizing gain of bit length polynomial in the
instance size.
The YES instances of this note are stabilized by the integer gains
$\hat\K\in\{1,2\}^{1\times N}$.
The third is completeness for the existential theory of the reals.
Deciding the solvability of a conjunction of strict polynomial inequalities is
polynomial-time equivalent to the full theory~\cite[Thm.~4.1]{SS:17}, so the
openness of the set of stabilizing gains is no obstruction to a reduction in
the converse direction.
The present construction realizes only two constraint types, membership of a
gain in a fixed union of two intervals and a two-sided bound on a weighted sum
of the gains.
A hardness proof for the existential theory would in addition need a block,
driven by three gains $\ki_a,\ki_b,\ki_c$, that is closed-loop stable if and
only if $\abs{\ki_a\ki_b-\ki_c}<\delta$, together with error control for the
coupled blocks analogous to Theorem~\ref{thm:localized}.
\section*{Declaration on generative AI}
OpenAI ChatGPT 5.6 and Claude Opus 5 were used for adversarial review and writing assistance during the writing of this  as well is writing the search scripts for finding the parameters of Definition \ref{def:schedule}. They were not used for idea generation nor proof writing. The authors assume responsibility for all the content.
\appendices
\section{Explicit construction of a prototype $\Hb$} \label{app:Hb}
We first solve the interpolation problem defining $\Hb$, and then verify the properties of Proposition~\ref{prop:Hb}. Throughout, $\eps\in(0,1/8]$. For
\[
P_\ki(p)=p^5+c_4p^4+c_3p^3+c_2p^2+c_1p+c_0,
\]
the coefficients are
\[
\begin{gathered}
c_4=d_4,\quad 
 c_3=d_3-\ki,\quad 
c_2=d_2+\ki n_2,\\
c_1=d_1+\ki n_1,\quad 
 c_0=d_0+\ki n_0.
\end{gathered}
\]
At a nonzero imaginary frequency, setting $y=\nu^2$ gives
\[
P_\ki(\J\nu) =c_4y^2-c_2y+c_0 +\J\nu(y^2-c_3y+c_1).
\]
Consequently, the interpolation conditions can be written as the following linear system with \emph{rational} coefficients:
\[
d_0+(1-\eps)n_0=0,
\]
\[
\begin{gathered}
d_4y^2-(d_2+g n_2)y+d_0+g n_0=0,\\
 y^2-(d_3-g)y+d_1+g n_1=0
\end{gathered}
\]
for
\[
(y,g)\in \{(1,1+\eps),(2,2-\eps),(3,2+\eps)\},
\]
and
\[
n_0+3n_1+9n_2=27.
\]
Order the unknowns as $(d_0,\dots,d_4,n_0,n_1,n_2)$ and the equations as displayed, taking the real equation before the imaginary equation for each pair $(y,g)$. The determinant of the coefficient matrix is $42(4\eps-1)^2$, which is nonzero on $(0,1/8]$. Thus the interpolation problem has a unique solution. 

Set
\[
\lambda\coloneqq \frac{2-11\eps}{1-4\eps}.
\]
Since $\lambda'(\eps)=-3/(1-4\eps)^2$, we have $\lambda\in[5/4,2)$. Solving the system gives
\[
N_b(p) =-p^3+\frac{8\lambda}{7}p^2 +\frac{8\eps+1}{1-4\eps}p+\frac{12\lambda}{7},
\]
and
\[
\begin{aligned}
D_b(p)={}&p^5+\frac{2(4\eps+1)\lambda}{7}p^4+\frac{7-21\eps-4\eps^2}{1-4\eps}p^3\\
&-\frac{6(2-11\eps)}{7}p^2 +\frac{4-23\eps-8\eps^2}{1-4\eps}p -\frac{12(1-\eps)\lambda}{7}.
\end{aligned}
\]
These expressions can also be verified by substitution into the eight equations. For rational $\eps$, they require a fixed number of rational arithmetic operations, and their numerators and denominators have bit length polynomial in that of $\eps$. This proves the asserted effective construction. The interpolation conditions prescribe four crossing gains. The following lemma shows that there are no others.

\begin{lemma}[Crossing gains]\label{lem:cross}
Let $\eps\in(0,1/8]$. The polynomial $P_\ki=D_b+\ki N_b$ has a root on the imaginary axis if and only if $\ki\in\{1-\eps,1+\eps,2-\eps,2+\eps\}$.
\end{lemma}
\begin{proof}
Since $c_0=d_0+\ki n_0=\frac{12\lambda}{7}\bigl(\ki-(1-\eps)\bigr)$ and
$\lambda>0$, the polynomial $P_\ki$ has a root at $p=0$ if and only if
$\ki=1-\eps$.
A root at $p=\J\nu$ with $\nu\neq0$ means, with $y\coloneqq \nu^2>0$,
\[
f(y)\coloneqq c_4y^2-c_2y+c_0=0
\quad\text{and}\quad
h(y)\coloneqq y^2-c_3y+c_1=0 ,
\]
because the real part of $P_\ki(\J\nu)$ equals $f(\nu^2)$ and the imaginary part
equals $\nu\,h(\nu^2)$.
A common root of $f$ and $h$ leads to $\Res_y(f,h)=0$.
A direct computation gives
\[\begin{aligned}
\Res_y(f,h)&=-\frac{8(4\eps+5)^2(11\eps-2)^2}{49(4\eps-1)^3}\\
&\quad\cdot(\eps-\ki+1)(\eps-\ki+2)(\eps+\ki-2).
\end{aligned}\]
The prefactors do not vanish on $(0,1/8]$.
Hence, if $P_\ki$ has a root $\J\nu$ with $\nu\neq0$, then $\Res_y(f,h)=0$,
and consequently $\ki\in\{1+\eps,\,2-\eps,\,2+\eps\}$.
Conversely, direct substitution gives $f(1)=h(1)=0$ at $\ki=1+\eps$,
$f(2)=h(2)=0$ at $\ki=2-\eps$, and $f(3)=h(3)=0$ at $\ki=2+\eps$.
So $P_\ki$ has the roots $\pm\J$, $\pm\J\sqrt2$ and $\pm\J\sqrt3$ at these three
gains, and it has the root $0$ at $\ki=1-\eps$.
\end{proof}
\begin{proof}[Proof of Proposition~\ref{prop:Hb}]
We first establish the stabilizing intervals in \ref{propitem:hurwitz2}. By Lemma~\ref{lem:cross}, the number of roots in the open right half-plane is constant on each connected component of
\[
\mathbb{R}\setminus\{1-\eps,1+\eps,2-\eps,2+\eps\}.
\]
Indeed, the coefficients depend continuously on $\ki$, and no root crosses the imaginary axis within a component. Since $P_\ki$ is monic of fixed degree, its roots remain
bounded on every compact gain interval. It therefore suffices to determine stability at one gain in each of the five components. We use Routh's criterion in the form of~\cite[Vol.~2, Ch.~XV, Sec.~3, p.~180]{Gan:59}. The first column of the Routh array is
\[
1,\quad c_4,\quad b_1,\quad b_2,\quad b_3,\quad c_0,
\]
where direct calculation gives
\[
\begin{gathered}
b_1=\frac{(4\eps+5)(\ki_c-\ki)}{4\eps+1},
\\
\ki_c\coloneqq \frac{2-5\eps-4\eps^2}{1-4\eps}
=2+\eps+\frac{2\eps}{1-4\eps},
\end{gathered}
\]
\[\begin{gathered}
b_2=\frac{8\lambda\,\eta(\ki)}{7(1-4\eps)(\ki-\ki_c)},\\
\begin{aligned}
\eta(\ki)&\coloneqq(1-4\eps)\ki^2+(12\eps-3)\ki\\
&\quad+(4\eps^3+\eps^2-10\eps+2),
\end{aligned}\end{gathered}\]
and
\[
b_3=
\frac{(4\eps+5)((1+\eps)-\ki)((2+\eps)-\ki)
      (\ki-(2-\eps))}
     {2\eta(\ki)}.
\]
Each expression is used only where the preceding pivots are nonzero. We have $c_4>0$, $\ki_c>2+\eps$, and
\[
\eta(1-\eps)=\eta(2+\eps)=-\eps(1+2\eps)<0.
\]
Since $\eta$ is a quadratic with positive leading coefficient, it follows that $\eta<0$ on $[1-\eps,2+\eps]$.

At $\ki=1$ and $\ki=2$, the pivots $b_1,b_2,c_0$ are positive, and the numerator of $b_3$ is negative. Thus $b_3>0$ and both polynomials are Hurwitz. At $\ki=0$, the constant coefficient
is negative, which excludes Hurwitz stability. At $\ki=3/2$, the pivots $b_1,b_2$ are positive and $b_3<0$, so the polynomial is not Hurwitz. Finally, at $\ki=\ki_c+1$, the pivot $b_1$ is negative, while the preceding pivots $1,c_4$ are positive, which again excludes Hurwitz stability.

The gains $0$, $1$, $3/2$, $2$, and $\ki_c+1$ belong to the five components, respectively. Constancy of the root count therefore identifies the two stabilizing components as $(1-\eps,1+\eps)$ and $(2-\eps,2+\eps)$. The four boundary gains have imaginary-axis roots by Lemma~\ref{lem:cross}, so none is Hurwitz. This proves \ref{propitem:hurwitz2}, including the assertion about imaginary-axis roots.

For \ref{propitem:nothurwitz}, we evaluate the same Routh array at $\ki=0$. Since
\[
\eta(0)=4\eps^3+\eps^2-10\eps+2>0,
\]
its first-column signs are
\[
+,\quad +,\quad +,\quad -,\quad -,\quad -.
\]
There are no imaginary-axis roots by Lemma~\ref{lem:cross}, and hence Routh's theorem gives exactly one root in the open right half-plane. Since nonreal roots occur in conjugate pairs, this root is real. Moreover,
\[
\begin{gathered}
D_b(0)=-\frac{12(1-\eps)\lambda}{7}<0,
\\
D_b(1)=
\frac{4(13-18\eps-142\eps^2)}{7(1-4\eps)}>0
\end{gathered}
\]
for $\eps\in(0,1/8]$, so the root lies in $(0,1)$.

For \ref{propitem:hurwitz}, define
\[
\zeta_1\coloneqq \frac{4\eps+5}{7(1-4\eps)}\ge\frac57, \qquad \zeta_0\coloneqq \frac{4\lambda}{7}>0.
\]
Expanding the proposed factorization gives
\[
-(p-3)(p^2+\zeta_1p+\zeta_0)
=-p^3+(3-\zeta_1)p^2+(3\zeta_1-\zeta_0)p+3\zeta_0.
\]
The identities
\[
3-\zeta_1=\frac{8\lambda}{7},\qquad
3\zeta_1-\zeta_0=\frac{1+8\eps}{1-4\eps},\qquad
3\zeta_0=\frac{12\lambda}{7}
\]
show that this polynomial equals $N_b$. The quadratic factor is Hurwitz because both coefficients are positive, so $p=3$ is the only numerator root in $\CRHP$ and is simple.

We next establish the coefficient bounds in \ref{propitem:coeff}, together with the individual estimates used in the margin calculation:
\begin{equation}\label{eq:coeffbounds}
\begin{gathered}
0<d_4\le6/7,\quad 7\le d_3\le9,\quad \abs{d_2}\le12/7,\\
0<d_1\le4,\quad \abs{d_0}\le24/7,\\
0<n_2\le16/7,\quad 0<n_1\le4,\quad 0<n_0\le24/7.
\end{gathered}
\end{equation}
The bounds on $d_4,n_2,n_0,d_0$ follow from $\lambda\in[5/4,2)$ and $4\eps+1\le3/2$. The bound on $d_2$ follows from $0<2-11\eps<2$. For $d_3$, the upper and lower bounds are equivalent to
\[
15\eps\le2+4\eps^2,\qquad 4\eps^2\le7\eps,
\]
respectively, and both hold on $(0,1/8]$. The numerator $4-23\eps-8\eps^2$ is positive on this interval, and $d_1\le4$ follows from $7\eps+8\eps^2\ge0$. Finally, $0<n_1\le4$ follows from $1+8\eps\le2$ and $1-4\eps\ge1/2$. Summing these estimates gives
\[
S_d\le\frac67+9+\frac{12}{7}+4+\frac{24}{7}=19,
\qquad
S_n\le\frac{16}{7}+4+\frac{24}{7}<10.
\]

It remains to prove \ref{propitem:noroots}. The disk $\{p:\abs{p-3}\le7/8\}$ lies in $\{p:\operatorname{Re}p\ge17/8\}$, whereas every root of $D_b$ is either in $\OLHP$ or in $(0,1)$ by \ref{propitem:nothurwitz}. Thus $D_b$ has no root in this disk. To obtain the lower bound on $\Hb$, consider $5/8\le\abs{p-3}\le7/8$. Then $\operatorname{Re}p\ge17/8$ and $\abs{p}\le31/8$. The discriminant of the quadratic numerator factor is
\[
\zeta_1^2-4\zeta_0 =-\frac{4912\eps^2-2168\eps+199}{49(1-4\eps)^2}<0.
\]
Indeed, the polynomial in the numerator is decreasing on $[0,1/8]$ and equals $19/4$ at $\eps=1/8$. The two quadratic roots therefore have real part $-\zeta_1/2\le-5/14$, and their distances from $p$ are each at least $17/8+5/14$. Hence
\[
\abs{N_b(p)} \ge\frac58\left(\frac{17}{8}+\frac{5}{14}\right)^2>\frac{15}{4}.
\]
The individual coefficient bounds give
\[
\begin{aligned}
\abs{D_b(p)} &\le \left(\frac{31}{8}\right)^5 +\frac{6}{7}\left(\frac{31}{8}\right)^4 +9\left(\frac{31}{8}\right)^3\\
&\quad +\frac{12}{7}\left(\frac{31}{8}\right)^2 +4\left(\frac{31}{8}\right)+\frac{24}{7} <1636.
\end{aligned}
\]
Consequently,
\[
\abs{\Hb(p)} >\frac{15}{6544}>\frac{1}{500},
\]
which completes the proof.
\end{proof}
\section{Estimates for frequency separation} \label{app:equiv}
We prove the estimates used in Section~\ref{sec:equiv}, with the notation and parameter ranges introduced there. The scalar decay bounds control the loop contributions away from their assigned bands, while the root bounds ensure that each comparison region contains all relevant roots of its scalar factor.
\begin{lemma}\label{lem:decay}
Let $p,s\in\mathbb{C}$. The following hold:
\begin{enumerate}[label=(\roman*)]
    \item If $\abs p\ge R$, then $\abs{\Hb(p)}\le 4/\abs p^{2}$.
        \label{lemitem:decp}
    \item If $\abs s\ge R_H$, then $\abs{\Hsum(s)}\le 4/\abs s$. 
        \label{lemitem:decs}
\end{enumerate}
\end{lemma}
\begin{proof}
    \ref{lemitem:decp} For $\abs p\ge R$ we have $\abs{N_b(p)}\le\abs p^3+S_n\abs p^2\le2\abs p^3$, since $S_n\le R\le\abs p$, and $\abs{D_b(p)}\ge\abs p^5-S_d\abs p^4\ge\abs p^5/2$, since $2S_d\le R\le\abs p$. Hence $\abs{\Hb(p)}\le4/\abs p^2$.

\ref{lemitem:decs} For $\abs s\ge R_H$ we have $\abs{s^2-s}\le2\abs s^2$, since $\abs s\ge2$, and
\[
\abs{\Dsum(s)}\ge\abs s^3\bigl(1-2\Wstar/\abs s^2-q_0/\abs s^3\bigr)\ge\abs s^3/2 ,
\]
since $R_H=2(1+\Wstar)^2$ and $\Wstar\ge1$ give $2\Wstar/R_H^2\le1/32$ and
$q_0/R_H^3\le1/512$.
Hence $\abs{\Hsum(s)}\le4/\abs s$.
\end{proof}
\begin{lemma}[Root location]\label{lem:loc}
Assume \ref{assitem:A2}.
For every $i\in\{1,\dots,N\}$ and every $\ki\in\mathbb{R}$ with
$\abs{\ki}\le\kappa$, all roots of $D_i$ and of $D_i+\ki N_i$ lie in the annulus
$\{2\omega_i/3\le\abs{s}\le3\omega_i/2\}$.
That annulus is contained in the interior of
$\{\omega_i/2\le\abs s\le2\omega_i\}$.
For every $\W\in\mathbb{R}$ with $\abs{\W}\le w$, all roots of $\Dsum$ and of
$\Dsum+\W\Nsum$ satisfy $\abs{s}\le R_W(w)\le\omega_1/4$.
\end{lemma}

\begin{proof}
By Proposition~\ref{prop:bpmap}\ref{propitem:bpmap2} the roots come in pairs $s,s'$ with
$ss'=\omega_i^2$ and $s+s'=B_wq$, where $q$ is a root of $D_b+\ki N_b$.
The Cauchy root bound in the form $\abs q\le1+\max_j\abs{a_j}$, over the
coefficients $a_j$ of a monic polynomial, gives
$\abs q\le1+S_d+\kappa(1+S_n)=Q(\kappa)$.
Then
$\abs{s}\le\tfrac{1}{2}\bigl(B_w\abs q+\sqrt{B_w^2\abs q^2+4\omega_i^2}\bigr)\le B_wQ(\kappa)+\omega_i\le\tfrac{3}{2}\omega_i$,
using $B_wQ(\kappa)\le\omega_1/4\le\omega_i/2$.
From $\abs s\,\abs{s'}=\omega_i^2$ and $\abs{s'}\le\tfrac32\omega_i$ we get
$\abs s\ge\tfrac23\omega_i$.
The cubic bounds are the same root bound applied to
$s^3+\W s^2+(2\Wstar-\W)s+q_0$, whose coefficient moduli are at most $w$,
$2\Wstar+w$ and $(\Wstar)^2$, together with \ref{assitem:A2}.
\end{proof}
\begin{lemma}\label{lem:tail}
Let $s\in\CRHP$. Items \ref{lemitem:tail1} and \ref{lemitem:tail2} hold under \ref{assitem:A2}.
\begin{enumerate}[label=(\alph*)]
\item\label{lemitem:tail1} If $\abs s\ge\omega_1/2$, then $\abs{\Hsum(s)}\le 8/\omega_1$.
\item \label{lemitem:tail2} If $1\le j\le N$ and either $\abs s\le\omega_j/2$ or $\abs s\ge2\omega_j$, then $\abs{\Hb^j(s)}\le 16B_w^2/(9\omega_j^2)$.
\end{enumerate}
Moreover, items \ref{lemitem:tail3}, \ref{lemitem:tail4} and \ref{lemitem:tail5}  hold under Assumption~\ref{ass:DC} when $\abs{\ki_j}\le\kappa$ for all $j$ and $\abs\W\le w$.
\begin{enumerate}[resume,label=(\alph*)]
\item \label{lemitem:tail3} If $s\in\mathcal{G}$, then $\sum_{j=0}^{N}\abs{L_j(s)}\le \frac{8\abs\W}{\omega_1}
+\frac{2N\kappa B_w^2}{\omega_1^2}$.
\item \label{lemitem:tail4} If $s\in\partial\mathcal{R}_i$ for some $0\le i\le N$, then $\sum_{j\neq i}\abs{L_j(s)}\le \frac{8\abs\W}{\omega_1}
+\frac{2N\kappa B_w^2}{\omega_1^2}$. 
\item \label{lemitem:tail5} If $s$ lies on a circular arc of $\partial\mathcal{R}_i$ for some $0\le i\le N$, then $\abs{L_i(s)}\le \frac{8\abs\W}{\omega_1}
+\frac{2N\kappa B_w^2}{\omega_1^2}$. 
\end{enumerate}
\end{lemma}

\begin{proof}
By \ref{assitem:A2}, $\abs s\ge\omega_1/2\ge2R_H$.
Lemma~\ref{lem:decay} gives $\abs{\Hsum(s)}\le4/\abs s\le8/\omega_1$, implying item \ref{lemitem:tail1}.

At $s=0$ we have $\Hb^j(0)=0$, since $N_j$ is divisible by $s^2$ and $D_j(0)\neq0$. So let $s\neq0$ and $p\coloneqq p_{\omega_j}(s)$. For any $s\neq0$, $\abs{s^2+\omega_j^2}\ge\abs{\,\abs s^2-\omega_j^2\,}$, so $\abs{p}\ge\abs{\,\abs s^2-\omega_j^2\,}/(B_w\abs s)$. If $\abs s\le\omega_j/2$ this gives $\abs{p}\ge\tfrac{3}{4}\omega_j^2/(B_w\abs s)\ge\tfrac{3}{2}\omega_j/B_w$. If $\abs s\ge2\omega_j$ it gives $\abs{p}\ge\tfrac{3}{4}\abs s/B_w\ge\tfrac{3}{2}\omega_j/B_w$. Both lower bounds are therefore at least $\tfrac32\omega_1/B_w$, which is at least $6R$ by \ref{assitem:A2}, hence at least $R$. For $\abs p\ge R$, Lemma~\ref{lem:decay} gives $\abs{\Hb(p)}\le4/\abs p^2$. Substituting the two lower bounds gives item \ref{lemitem:tail2}.

For items \ref{lemitem:tail3} and \ref{lemitem:tail4}, note first that \ref{lemitem:tail2} and $\abs{\ki_j}\le\kappa$ give, wherever \ref{lemitem:tail2} applies, $\abs{L_j(s)}=\abs{\ki_j}\abs{\Hb^j(s)}\le2\kappa B_w^2/\omega_1^2$, that \ref{lemitem:tail1} gives $\abs{L_0(s)}=\abs\W\abs{\Hsum(s)}\le8\abs\W/\omega_1$ wherever \ref{lemitem:tail1} applies. A sum of at most $N$ band terms and one summing term is therefore bounded by
\[
\frac{8\abs\W}{\omega_1}+\frac{2N\kappa B_w^2}{\omega_1^2}.
\]
It remains to verify that the corresponding estimates apply
on each of the stated sets.

Let $s\in\mathcal{G}$. Then $\abs s\ge2\omega_1\ge\omega_1/2$, so \ref{lemitem:tail1} applies to $L_0$. For $1\le j\le N$, the definition of $\mathcal{G}$ gives $\abs s\le\omega_j/2$
or $\abs s\ge2\omega_j$, so \ref{lemitem:tail2} applies to $L_j$. This proves \ref{lemitem:tail3}.

Let $s\in\partial\mathcal{R}_i$ with $i\ge1$. Then $\omega_i/2\le\abs s\le2\omega_i$, so \ref{lemitem:tail1} applies to $L_0$. For $j\ge1$ with $j\neq i$, \ref{assitem:A1} gives $\omega_j\le\omega_i/4$ or $\omega_j\ge4\omega_i$, hence $\abs s\ge2\omega_j$ or $\abs s\le\omega_j/2$, so \ref{lemitem:tail2} applies to $L_j$. Let $s\in\partial\mathcal{R}_0$. Then $\abs s\le\omega_1/2\le\omega_j/2$ for every $j\ge1$, so \ref{lemitem:tail2} applies to every band term. This proves \ref{lemitem:tail4}.

For $i\ge1$ the arcs are $\abs s=\omega_i/2$ and $\abs s=2\omega_i$, where \ref{lemitem:tail2} and $\abs{\ki_i}\le\kappa$ give $\abs{L_i(s)}\le2\kappa B_w^2/\omega_1^2$. For $i=0$ the arc is $\abs s=\omega_1/2$, where \ref{lemitem:tail1} gives $\abs{L_0(s)}\le8\abs\W/\omega_1$. Combining both proves the bound in \ref{lemitem:tail5}.
\end{proof}

\bibliographystyle{IEEEtran}
\bibliography{MyBib}
\end{document}